\documentclass[12pt, a4paper, abstracton]{scrreprt}
\usepackage[T1]{fontenc}
\pdfoutput=1
\usepackage[utf8]{inputenc}
\usepackage{amsmath}
\usepackage{amssymb}
\usepackage{amsthm}
\usepackage{microtype}
\usepackage{lmodern}
\usepackage{units}
\usepackage{graphicx}
\usepackage{enumitem}
\usepackage{subfigure}
\usepackage{paralist}
\usepackage{schuster}
\usepackage{mathtools}
\usepackage{hyperref}
\usepackage[sort,comma,square,numbers]{natbib}
\usepackage{scrpage2}
\usepackage[compact]{titlesec}
\usepackage[figure]{hypcap}

\renewcommand{\cite}{~\citep}

\hypersetup{breaklinks,a4paper,linktocpage,bookmarksopen, bookmarksopenlevel=2, pdftitle={A Quadratic-Vertex Problem Kernel for s-Plex Cluster Vertex Deletion}, pdfauthor={Ren\'e van Bevern}, pdfsubject={Graph-Modeled Data Clustering}, pdfkeywords={data clustering, graph algorithms, problem kernel, fixed-parameter algorithmics},pdfdisplaydoctitle,colorlinks,linkcolor=blue,citecolor=blue}

\setkomafont{sectioning}{\normalfont\bfseries}
\setkomafont{descriptionlabel}{\normalfont\bfseries}
\setkomafont{title}{\normalfont}
\setcapindent{1em}

\newcommand{\premv}[1]{Let~$X$ be a solution. For each vertex~$v\in X$, let~$M(v)$ be the set constructed by \autoref{#1}.}
\newcommand{\prem}[1]{Let~$X$ be a solution. Let~$M:=\bigcup_{v\in X}M(v)$ be the set constructed by \autoref{#1}.}
\newcommand{\N}{\mathbb{N}}

\newcommand{\DP}{\text{P}}
\newcommand{\md}[1]{\mbox{$#1\!$-module}}
\newcommand{\nx}[1]{N(#1)\cap X}

\newcommand{\NP}{\text{NP}}

\newcommand{\name}{\textsc}
\newcommand{\hneg}{\mathcal H_0(X,M)}
\newcommand{\hpos}{\mathcal H_1(X,M)}
\newcommand{\hv}{\mathcal H(X)}
\newcommand{\hp}{H}

\newcommand{\pvd}[1]{\name{\mbox{$#1$-Plex} Cluster Vertex Deletion}}
\newcommand{\cvd}{\name{Cluster Vertex Deletion}}
\newcommand{\hs}[1]{\name{\mbox{$#1$-Hitting} Set}}
\newcommand{\pl}[1]{\mbox{$#1$-plex}}
\newcommand{\pcg}[1]{\pl #1 cluster graph}
\newcommand{\FISG}{\textsc{Fisg}}

\newcommand{\decprob}[3]{%
\begin{flushright}
\addtolength{\linewidth}{-1em}
\begin{minipage}{\linewidth}
\textsc{#1}
\begin{compactdesc}
\item[\hspace{\parindent}Instance:] #2
\item[\hspace{\parindent}Question:] #3
\end{compactdesc}  
\end{minipage}
\addtolength{\linewidth}{1em}
\end{flushright}
}

\newtheorem{satz}{Theorem}[chapter]
\newtheorem{korollar}{Corollary}[chapter]
\newtheorem{lemma}{Lemma}[chapter]

\theoremstyle{definition}
\newtheorem{rul}{Reduction Rule}[chapter]
\newtheorem{asu}{Assumption}[chapter]
\newtheorem{definition}{Definition}[chapter]
\newtheorem{proc}{Algorithm}[chapter]

\theoremstyle{remark}

\begin{document}
\mathtoolsset{centercolon}
\begin{titlepage}
\thispagestyle{scrheadings}
\cfoot{}
\setfootsepline{0.5pt}
\ifoot{\vspace{1cm}\\
Cite this work as: René van Bevern. \emph{A quadratic-vertex problem kernel for $s$-plex cluster vertex deletion}. Studienarbeit, Department of Mathematics and Computer Science, Friedrich Schiller University of Jena, Germany, 2009.}
\centering
\noindent \rule{\textwidth}{0.5pt}

\vspace{1em}
\Large Graph-Based Data Clustering:\\\Huge A Quadratic-\!Vertex Problem Kernel for\\$s$-Plex Cluster Vertex Deletion

\normalsize\vspace{1.12\topsep}\large
by

\normalsize\vspace{\topsep}\Large
\textsc{Ren\'e van Bevern}\\
\rule{\textwidth}{0.5pt}

\vfill
\Large Studienarbeit\\
\normalsize 14. September 2009

\vspace{\topsep}
  Betreuung: Hannes Moser,\\
  Rolf Niedermeier      

\vfill
\normalfont
\includegraphics[width=5cm,trim=0cm 4cm 0cm 0cm,clip=true]{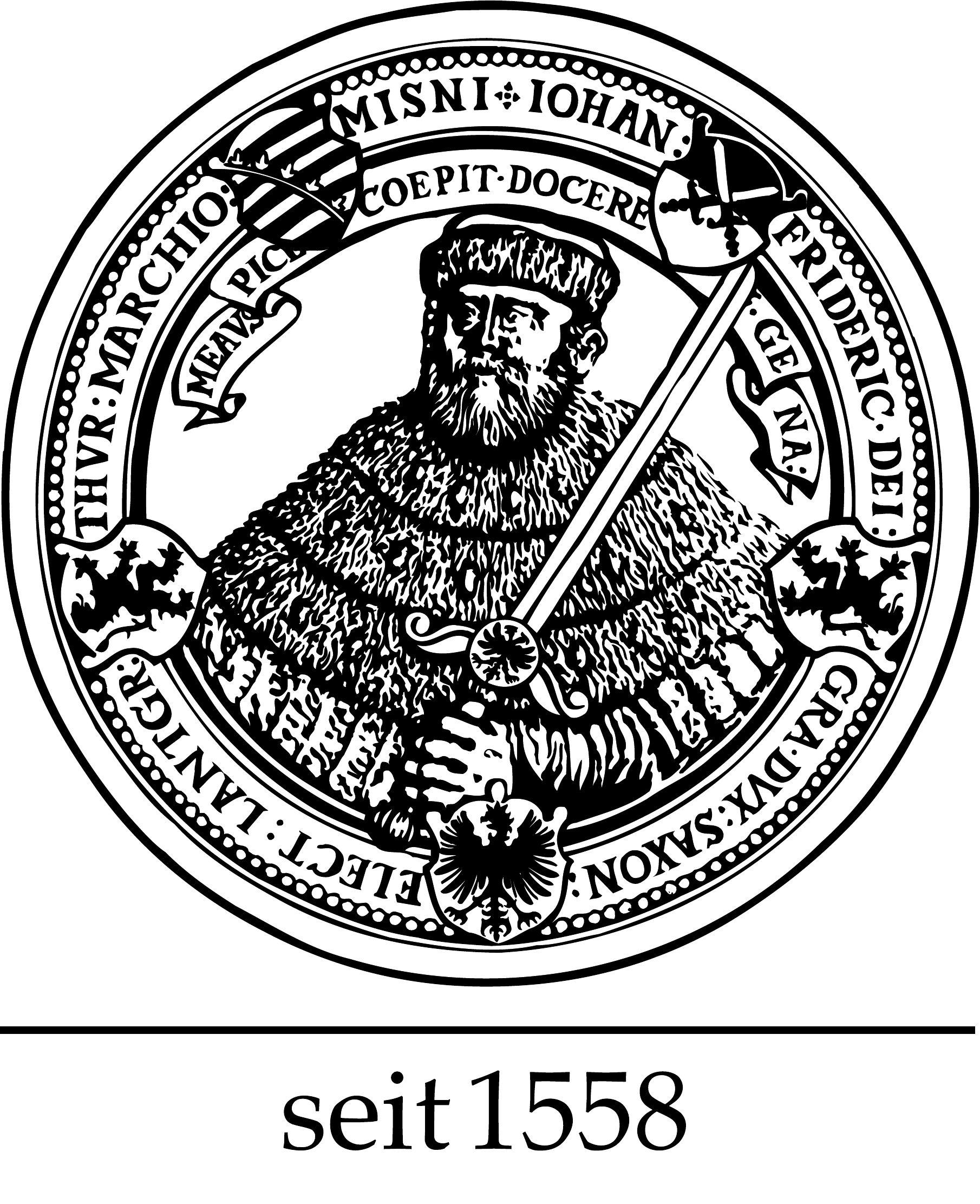}

\vspace{\topsep}
Friedrich-Schiller-Universität Jena\\
  Institut für Informatik\\
  Theoretische Informatik I / Komplexitätstheorie
\end{titlepage}

\begin{abstract}
  We introduce the \pvd s problem. Like the \cvd{} problem, it is NP-hard and motivated by graph-based data clustering. While the task in \cvd{} is to delete vertices from a graph so that its connected components become \emph{cliques}, the task in \pvd s is to delete vertices from a graph so that its connected components become \emph{\pl ses}. An \pl s is a graph in which every vertex is nonadjacent to at most~$s-1$ other vertices; a clique is an \pl1. In contrast to \cvd, \pvd s allows to balance the number of vertex deletions against the sizes and the density of the resulting clusters, which are \pl ses instead of cliques.

  The focus of this work is the development of provably efficient and effective data reduction rules for \pvd s. In terms of fixed-parameter algorithmics, these yield a so-called \emph{problem~kernel}. A similar problem, \textsc{$s$-Plex Editing}, where the task is the insertion or the deletion of edges so that the connected components of a graph become \pl ses, has also been studied in terms of fixed-parameter algorithmics. Using the number of allowed graph modifications as parameter, we expect typical parameter values for \pvd s to be significantly lower than for \textsc{$s$-Plex Editing}, because one vertex deletion can lead to a high number of edge deletions. This holds out the prospect for faster fixed-parameter algorithms for \pvd s.
\end{abstract}

\tableofcontents

\chapter{Introduction}\label{introsec}
Data clustering problems are of great importance in the disciplines of machine learning, pattern recognition, and data mining\cite{Berkhin2006}.  Given a data set, one can define a measure of similarity on data pairs. The goal in data clustering is to partition the data set into \emph{clusters} so that the elements within a cluster are similar, while there are less similarities between vertices in different clusters. Mapping clustering tasks into graph-theoretic models allows the usage of the broad variety of graph algorithms to process and cluster data\cite{DBLP:journals/csr/Schaeffer07}. Usually, the similarity between data records is mapped to a graph~$G$ as follows: each vertex in~$G$ corresponds to a data record, and an edge between two vertices in~$G$ exists if and only if the similarity of the corresponding data records exceeds a certain threshold. This threshold is specific to the actual clustering problem. An obvious possible postulation on clusters is for each data pair in one cluster to be similar. A cluster can therefore be interpreted as a complete graph, also called \emph{clique}. Subject to our goal that there shall be only few similarities between vertices in different clusters, the graph~$G$ constructed from our data would ideally consist of isolated cliques only. Such a graph is called a \emph{cluster graph}. For real-world data, it is unrealistic to expect~$G$ to be a cluster graph. We could modify~$G$ to become a cluster graph, but because we want to avoid excessive perturbation of the input data, the graph should be modified only modestly. One way to model this task is \cvd\cite{HKMN09TOCS}.  \decprob{\cvd}{An undirected graph~$G=(V,E)$ and a natural number~$k$.}{Is there a vertex set~$S\subseteq V$ with~$|S|\leq k$ such that deleting all vertices in~$S$ from~$G$ results in a graph where each connected component forms a clique?}

\noindent This problem corresponds to discarding at most $k$~data records in order to find a plausible data clustering. We can regard the discarded data records as outliers. Although \cvd{} is a very intuitive model of graph-based data clustering, it is very restrictive as it requires \emph{every} data pair in a cluster to be similar.  \cvd{} offers no option to relax this requirement, so that we could allow for a few dissimilarities within the resulting clusters. Obviously, it is desirable to balance the amount of discarded data against the number of dissimilarities within a cluster. Also, inaccuracies in the data could render finding satisfactory clustering results using \cvd{} impossible, yielding too many or too small clusters. Therefore, we weaken the requirement for every connected component to form a clique. Seidman and Foster\cite{SF78} have introduced one generalization of the clique concept in 1978:

\begin{definition}
  For~$s\geq 1$, an \emph{$s$-plex} is a graph~$G=(V,E)$ such that every
  vertex in~$V$ is adjacent to at least $|V|-s$~other vertices in~$V$.
\end{definition}

\noindent For example, a clique is an \pl1. By modeling clusters using \pl{s}es instead of cliques, we allow each data record to be dissimilar to $s-1$ other data records within the same cluster.  Although the \pl s concept has already been introduced in 1978, it has only recently become subject to algorithmic research\cite{BBH09,DBLP:conf/aaim/GuoKNU09,MH09,MNS09,WP07}. In this work, we introduce the \pvd s~problem.

\decprob{\pvd s}{An undirected
  graph~$G=(V,E)$ and a natural number~$k$.}{Is there a vertex
  set~$S\subseteq V$ with~$|S|\leq k$ such that deleting all vertices
  in~$S$ from~$G$ results in a graph where each connected component
  forms an \pl s?}

\noindent In the following, we will call a graph that has only \pl ses
as connected components an \emph{\pcg s}.  For each~$s$, the \pvd s
problem yields a different clustering model. In each model,
$s$~determines the ``density'' of the resulting clusters and with that
the dissimilarities that are allowed within each cluster.

\paragraph{Fixed-Parameter Algorithmics.}
In this work, we study the \pvd s problem in terms of fixed-parameter algorithmics. Fixed-parameter algorithmics aims at a multivariate complexity analysis of problems without giving up the demand for finding optimal solutions\cite{DF99,Flu06,Nie06}. A \emph{parameterized} problem is a language~$L\subseteq \Sigma^{*}\times \N$, where~$\Sigma$~is a finite alphabet. The second component is called the \emph{parameter} of the problem. The \pvd s problem is a parameterized problem with the input~$G$ and the parameter~$k$. A parameterized problem~$L$ is \emph{fixed-parameter tractable} if it can be determined in $f(k) |x|^{O(1)}$~time whether~$(x,k)\in L$, where~$f$ is a computable function only depending on~$k$. The corresponding complexity class is called \emph{FPT}.

Given a parameterized problem instance~$(x,k)$, \emph{reduction to a problem kernel} or \emph{kernelization} means to transform~$(x,k)$ into an instance~$(x',k')$ in polynomial time, such that the size of~$x'$ is bounded from above by some function only depending on~$k$, ~$k'\leq k$, and~$(x,k)$ is a yes-instance if and only if~$(x',k')$ is a yes-instance. We refer to~$(x',k')$ as \emph{problem kernel}. Kernelization enables us to develop provably efficient and effective \emph{data reduction rules}. Refer to Guo and Niedermeier\cite{GN07SIGACT} for a survey on problem kernelization. In this work, we present a kernelization for \pvd s.

\paragraph{Terminology.}
\label{sec:prelim}
We only consider \emph{undirected} graphs~$G=(V,E)$, where~$V$ is the set of vertices and~$E$ is the set of edges. Throughout this work, we use $n:=|V|$ and~$m:=|E|$. We call two vertices~$v,w\in V$ \emph{adjacent} or \emph{neighbors} if~$\{v,w\}\in E$. The \emph{neighborhood}~$N(v)$ of a vertex~$v\in V$ is the set of vertices that are adjacent to~$v$.  For a vertex set~$U\subseteq V$, we set~$N(U):=\bigcup_{v\in U} N(v)\setminus U$. We call a vertex~$v\in V$ \emph{adjacent} to~$V'\subseteq V$ if $v$~has a neighbor in~$V'$. Analogously, we extend this definition and call a vertex set~$U\subseteq V$ \emph{adjacent} to a vertex set~$W\subseteq V$ with~$W\cap U=\emptyset$ if~$N(U)\cap W\ne\emptyset$. A \emph{path} in~$G$ from~$v_1$ to~$v_\ell$ is a sequence~$(v_1,v_2,\dots,v_\ell)\in V^{\ell}$ of vertices with~$\{v_i,v_{i+1}\}\in E$ for~$i\in\{1,\dots,\ell-1\}$. We call two vertices~$v$~and~$w$ \emph{connected} in~$G$ if there exists a path from~$v$ to~$w$ in~$G$.  For a set of vertices~$V' \subseteq V$, the \emph{induced subgraph}~$G[V']$ is the graph over the vertex set~$V'$ with the edge set~$\{\{v, w\} \in E \mid v, w \in V'\}$.  For~$V'\subseteq V$, we use $G-V'$ as an abbreviation for~$G[V\setminus V']$.

\paragraph{Related Work.}
The two ``sister problems'' of \pvd s, namely \name{$s$-Plex Editing} and \cvd, have been subject to recent research\cite{DBLP:conf/aaim/GuoKNU09,HKMN09TOCS}. The goal of the \name{$s$-Plex Editing} problem is to transform a graph into an~$s$-plex cluster graph by insertion or removal of at most $k$~\emph{edges}. For \cvd, Hüffner et al.\cite{ HKMN09TOCS} have developed fixed-parameter algorithms using the recent iterative compression\cite{GMN2009} technique introduced by Reed et al.\cite{RSV04}. Their algorithm solves \cvd{} in $O(2^k\cdot n^2(m+n\log n))$~time, where~$k$ is the number of allowed vertex deletions.  Guo~et~al.\cite{DBLP:conf/aaim/GuoKNU09} have shown a problem kernel with $O(ks^2)$~vertices for \name{$s$-Plex Editing}, where~$k$ is the number of allowed edge modifications. They also have developed the following forbidden induced subgraph characterization for \pcg ss.

\begin{satz}[Guo et al.\cite{DBLP:conf/aaim/GuoKNU09}]\label{fisg-char}
  Let~$G=(V,E)$ be a graph. Let $F$ be the set of all connected graphs with at most $|V|$ vertices that contain a vertex that is nonadjacent to~$s$ other vertices. The graph~$G$ is an \pcg s if and only if it does not contain any graph from $F$ as induced subgraph.
\end{satz}
\begin{figure}
  \centering
    \includegraphics{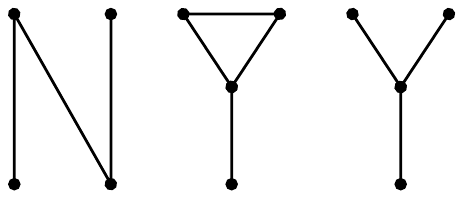}
    \caption{Minimal forbidden induced subgraphs for $s=2$.}
  \label{forbidden}
\end{figure}

\noindent
Guo et al.\cite{DBLP:conf/aaim/GuoKNU09} have also shown the stronger result that, for each natural number~$s$, there exists a natural number ${d\in O(s+\sqrt s)}$ such that if a graph~$G$ is not an \pcg s, then~$G$ contains a forbidden induced subgraph (\FISG{}) with at most~$d$ vertices. They present an algorithm that, if~$G$ is not an \pcg s, finds such a \FISG{} in~$G$ in~${O(s(n+m))}$~time. If~$s=2$ and if~$G$ is not a \pcg 2, then their algorithm always finds one of the three \FISG{}s shown in \autoref{forbidden}.  We can solve \pvd s by repeatedly finding a \FISG{} with at most $d$~vertices in $O(s(n+m))$~time and then branching into all possibilities of deleting one of its vertices. This yields a trivial search tree algorithm to solve \pvd s in ${O(d^k s(n+m))}$~time. Algorithms with a lower exponential time term can be obtained employing the \hs d problem:

\decprob{\hs d}{A set $H$, a collection of subsets $\mathcal C\subseteq\{H'\subseteq H\mid |H'|\leq d\}$ and a natural number~$k$.} {Is there a \emph{hitting set}~$S\subseteq H$ with~$|S|\leq k$ such that each set in $\mathcal C$ contains an element of~$S$?}

\noindent We obtain a \hs d instance~$(H,\mathcal C,k)$ from an \pvd s instance~$(G,k)$ as follows: we use the vertex set of~$G$ as~$H$; for each \FISG{}~$F$ containing at most~$d$ vertices from~$G$, we add the vertex set of~$F$ to~$\mathcal C$. Because each element in~$\mathcal C$ corresponds to a \FISG{} with at most~$d$ vertices, we have~$|\mathcal C|\in O(n^d)$. Because this bound is exponential in~$d$, it is practically infeasible to transform an \pvd s instance into a \hs d instance without prior data reduction. We can solve \hs d using a trivial $O(d^k|\mathcal C|)$-time search tree algorithm; we repeatedly choose a set from the collection $\mathcal C$ and branch into all possibilities of adding one of its vertices to a hitting set. Faster algorithms for \hs d are known\cite{Nie06}. For example, consider the special case~$s=2$. The \FISG{}s for \pvd 2 are shown in \autoref{forbidden}. The trivial search tree algorithm for \pvd 2 (as discussed above) runs in $O(4^k(n+m))$ time. We can solve an equivalent \hs 4 instance in $O(3.076^k+|\mathcal C|)$ time by combining Wahlström's $O(2.076^k+|\mathcal C|)$-time algorithm for \hs 3\cite{Wah07} with iterative compression, as discussed by Dom et al.\cite{DGHNT09}.

The forbidden induced subgraph characterization by Guo~et~al.\cite{DBLP:conf/aaim/GuoKNU09} implies that every induced subgraph of an~$s$-plex cluster graph is again an~$s$-plex cluster graph. The property of being an~$s$-plex cluster graph is thus \emph{hereditary}. Lewis and Yannakakis\cite{LY80} have shown that vertex deletion problems for hereditary graph properties are NP-hard. Because it can be verified in polynomial time whether a graph contains a \FISG{} for \pcg ss, \pvd s is in NP.  As a consequence, we can conclude that \pvd s is NP-complete. Further, Lund and Yannakakis\cite{CM93} have shown that vertex deletion problems for hereditary graph properties are constant-factor approximable and MAX SNP-hard, if the graph property admits a characterization by a finite number of \FISG{}s. Because \pcg ss are characterized by a finite number of \FISG{}s, finding a minimum solution for \pvd s is constant-factor approximable and MAX SNP-hard.

\paragraph{Our contributions.}
We show a problem kernel with $O(k^2)$~vertices for \pvd2, which can be found in $O(kn^2)$~time. We then generalize this kernelization algorithm to show a problem kernel with $O(k^2s^3)$~vertices for \pvd s, which can be found~in~$O(ksn^2)$~time.

\chapter{Kernelization for 2-Plex Cluster Vertex Deletion}
\label{2pvd}
In this chapter, we transform a \pvd2 instance~$(G,k)$ into a problem kernel~$(G',k')$. To this end, we present a series of data reduction rules that remove vertices from~$G$ so that the maximum number of vertices in the resulting graph~$G'$ depends only on the parameter~$k$. These data reduction rules also compute the new parameter~$k'\leq k$. For each data reduction rule, we show that it can be carried out in polynomial time and that it is \emph{correct}, that is, we show that $(G,k)$ is a yes-instance if and only if~$(G',k')$ is a yes-instance.

Assume that we are given a \pvd2 instance~$(G,k)$. We want to apply a series of data reduction rules to~$G$ so that we can bound the size of~$G$ by a function only depending on the parameter~$k$. To structure the graph~$G$, we first search for a constant-factor approximate solution~$X$ so that each connected component in~$G-X$ is a \pl2. This partitions the graph as shown in \autoref{packingfig}. To bound the overall size of~$G$ by a function only depending on the parameter~$k$, we independently bound the sizes of~$G-X$ and~$X$ by functions only depending on~$k$.

\begin{figure}[h]
  \centering
  \includegraphics{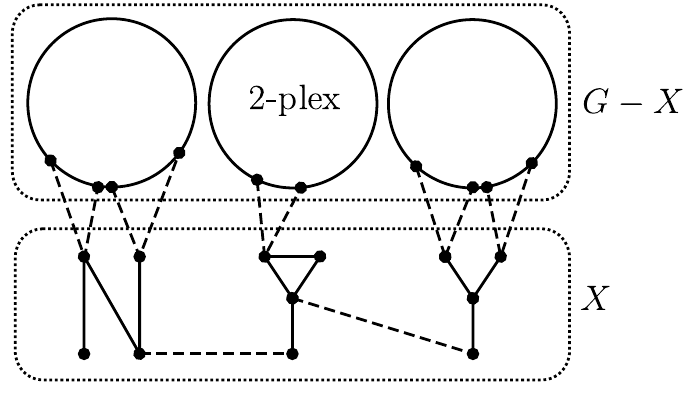}
  \caption{Constant-factor approximate solution $X$ and the graph~$G-X$.}
  \label{packingfig}
\end{figure}

\noindent To bound the size of~$X$, we use that~$X$ is a constant-factor approximate solution. If~$(G,k)$ is a yes-instance, then~$G$ can be transformed into a \pcg 2 by at most~$k$ vertex deletions. This implies that the size of~$X$ is at most~$ck$ for some constant factor~$c$. In particular, the maximum size of~$X$ only depends on~$k$. If~$X$ contains more than~$ck$ vertices, we stop our kernelization algorithm and output that~$(G,k)$ is a no-instance.

It is left to bound the size of~${G-X}$ by a function only depending on the parameter~$k$. To this end, we present data reduction rules to independently bound the \emph{number} and the \emph{sizes} of the connected components in~${G-X}$ by functions only depending on~$k$. Bounding the sizes of the connected components is the most sophisticated part of our kernelization algorithm. To this end, we employ graph separators and introduce a generalization of the graph module concept\cite{Gal67,DBLP:journals/dm/McConnellS99} in \autoref{redundant}.

Summarizing, we obtain a problem kernel for a \pvd 2 instance~$(G,k)$ by executing the following steps:
\begin{enumerate}
\item Find a constant-factor approximate solution~$X$ such that~$G-X$ is a \pcg 2. This is the subject of \autoref{greedy}. Because~$X$ is a constant-factor approximate solution, the size of~$X$ is bounded by a function only depending on the parameter~$k$. 
\item Bound the \emph{number} of connected components in~$G-X$ by a function only depending on the parameter~$k$. To this end, we use data reduction rules presented in \autoref{mark}.
\item Bound the \emph{sizes} of the connected component in~$G-X$ by a function only depending on the parameter~$k$. To this end, we use data reduction rules presented in \autoref{redundant}.
\end{enumerate}
In \autoref{result}, we show that the remaining graph (consisting of the vertices in~$X$ and the connected components in~$G-X$ to which all data reduction rules have been applied) contains $O(k^2)$ vertices. Together with the new parameter computed by our data reduction rules, this graph constitutes our problem kernel.

In the following, we write \emph{solution} for a vertex set~$X$ such that~$G-X$ is a \pcg 2. If we intend to refer to a solution containing at most~$k$ vertices, then we state it explicitly.

\section{An Approximate Solution}
\label{greedy}
In this section, we present an algorithm that greedily computes an approximate solution for \pvd 2. Given a graph~$G$, Guo~et~al.\cite{DBLP:conf/aaim/GuoKNU09} have shown that if~$G$ is not an \pcg s, an $O(s+\sqrt s)$-vertex \FISG{} in~$G$ can be found in $O(s(n+m))$ time. For the case~$s=2$, this algorithm finds the \FISG{}s shown in \autoref{forbidden}. We apply their algorithm for~$s=2$ to construct an initial solution:

\begin{proc}\label{find-X}
  Given a graph~$G$, we start with $H=G$ and $X=\emptyset$. We repeatedly apply the algorithm by Guo et~al.\cite{DBLP:conf/aaim/GuoKNU09} to find a \FISG{} in~$H$, we add its vertices to~$X$, and remove them from~$H$. If no \FISG{} can be found, then the algorithm stops and returns~$X$.
\end{proc}

\noindent \autoref{packingfig} illustrates the separation of~$G$
into~$X$ and~$H=G-X$.
\begin{lemma}\label{find-X-time}
  \autoref{find-X} computes a factor-4 approximate solution for \pvd 2. It can be carried out in $O(n(n+m))$~time.
\end{lemma}
\begin{proof}
  First, we show the running time.
  In each step, a \FISG{} can be found in ${O(n+m)}$ time. Because in each step of \autoref{find-X} four vertices are removed from~$H$, we apply it at most $O(n)$~times. Therefore, \autoref{find-X} runs in $O(n(n+m))$~time.

  It is left to show that the set~$X$ computed by \autoref{find-X} is a factor-4 approximate solution.  \autoref{find-X} stops when no more \FISG{}s can be found in~$H=G-X$. Thus, $H$ must be a 2-plex cluster graph and $X$ is a solution.

  Let $F$ be the set of all \FISG{}s found by \autoref{find-X}. Because each \FISG{} is deleted from~$H$ when it is discovered, the graphs in~$F$ are pairwise vertex-disjoint. Any solution must contain at least one vertex of each \FISG{} in~$F$. Therefore, the size of a solution is at least~$|F|$. Each \FISG{} found by the algorithm of Guo et al.\cite{DBLP:conf/aaim/GuoKNU09} contains four vertices. It follows that the solution $X$ computed by \autoref{find-X} contains~$4|F|$ vertices, which is at most four times the number of vertices in an optimal solution.
\end{proof}
\begin{korollar}\label{X-boundary}
  Let $(G,k)$ be a yes-instance. Then, \autoref{find-X} computes a solution for~$G$ that contains at most~$4k$ vertices.
\end{korollar}
\noindent Many of the following observations and data reduction rules require an initial solution~$X$. In those observations, we make no assumptions about~$X$ other than $X$~being a solution. For practical considerations, a heuristic search for an initial solution might be superior to employing \autoref{find-X}. Heuristic search might not only be faster, but might also find a smaller solution. This is desirable because the size of our problem kernel is proportional to the size of the initial solution. However, to conclude a problem kernel with~$O(k^2)$ vertices, we require an initial constant-factor approximate solution.

\section{Bounding the Number of Connected Components}
\label{mark}
Let~$X$ be a solution for~$G$. In this section we bound the number of connected components in~$G-X$ by a function only depending on the parameter~$k$. To this end, we employ a data reduction rule that resembles Buss and Goldsmith's\cite{BG93} kernelization of the \name{Vertex Cover} problem.

\begin{lemma}\label{high-occurence}
  Let~$(G,k)$ be a \pvd 2 instance and let $F(v)$ be a set of \FISG{}s pairwisely intersecting only in the vertex $v$ of~$G$. If $|F(v)|>k$, then $(G,k)$~is a yes-instance if and only if $(G-\{v\},k-1)$~is a yes-instance.
\end{lemma}

\begin{proof}
  If $(G,k)$~is a yes-instance, then there exists a solution~$S$ with~$|S|\leq k$ such that $G-S$~is a 2-plex cluster graph. The set $S\setminus\{v\}$~is a solution for~$G-\{v\}$.  If~$S$ does not contain~$v$, then it contains at least one vertex for every \FISG{} in~$F(v)$. Because there are more than~$k$ \FISG{}s in~$F(v)$, this contradicts~$|S|\leq k$.  Therefore, $v\in S$ and~$S\setminus\{v\}$ contains at most~$k-1$ vertices. This shows that~$(G-\{v\},k-1)$ is a yes-instance.

  If $(G-\{v\},k-1)$~is a yes-instance, then $G-\{v\}$~admits a solution~$S$ of size~$k-1$.  The set $S\cup \{v\}$ is a solution for~$G$ that contains at most $k$~vertices. Thus, $(G,k)$~is a yes-instance.
\end{proof}

\noindent In \autoref{mark-vertices}, we introduce the concept of \emph{peripheral} sets. Given a solution~$X$, peripheral sets help us in \autoref{unmarked-bounds} to bound the number of connected components in~$G-X$ and help us in \autoref{redundant} to bound their sizes. We present an algorithm that constructs a peripheral set efficiently and enables us to give a lower bound on the number of vertices that pairwisely intersect only in a single vertex~$v\in X$. If more than~$k$ \FISG{}s intersect only in~$v$, then we can remove~$v$ from~$G$ according to \autoref{high-occurence}.

\subsection{Peripheral Sets}\label{mark-vertices}
In this section, we present an algorithm that, for each vertex~$v$ in a solution~$X$, constructs a vertex set~$M(v)$ that allows us to give a lower bound on the number of \FISG{}s that pairwisely intersect only in the vertex~$v$. If this lower bound shows that more than~$k$ \FISG{}s pairwisely intersect only in~$v$, then we can remove~$v$ from~$G$ according~to~\autoref{high-occurence}.

As a side effect, we construct the sets~$M(v)$ so that their union~$M:=\bigcup_{v\in X}M(v)$ helps us to bound the number and the sizes of the connected components in~$G-X$: informally speaking, if we remove~$M$ from~$G$, then we want each vertex~$v\in X$ to be adjacent to only one large connected component in~$G-(X\cup M)$. As a result, there will be at most~$|X|$ large connected components in~$G-(X\cup M)$ adjacent to~$X$. Further, if a vertex~$v\in X$ has a neighbor in a connected component in~$G-(X\cup M)$, then we want the vertex $v$ to be adjacent to almost all of that connected component's vertices. This will help us in \autoref{redundant} to bound the sizes of the connected components in~$G-X$. We later formalize these properties and capture them under the concept~of~a~\emph{peripheral}~set.

We will see that we can easily bound the size of~$M$ by a function only depending on the parameter~$k$. Thus, the graph~$G-M$ can be thought of as the ``core'' of our kernelization problem, for which we must provide further data reduction rules. In contrast, the vertices in~$M$ are only of peripheral interest.

Given a solution~$X$ for~$G$, we now construct the set~$M(v)$ for each vertex~$v\in X$. We start with~$M(v)=\emptyset$. Then, we repeatedly search for a \FISG{}~$F$ in~$G$ that contains~$v$ but no vertices from~$M(v)$ and add the vertices of~$F-\{v\}$ to~$M(v)$. This ensures that we only find \FISG{}s that pairwisely intersect only in~$v$. To find such \FISG{}s, we present three observations on the connected components in~$G-X$. Each observation will lead to a phase of an algorithm that constructs the sets~$M(v)$.

\begin{definition}
  Let~$V$ be the vertex set of~$G$ and let~$X$ be a solution. We define the collection $\hv:=\{\hp\subseteq V\mid \hp\text{ induces a connected component in }G-X\}$ of the vertex sets of the connected components in~$G-X$.
\end{definition}
\noindent Because each set in~$\hv$ induces a connected component in~$G-X$ and because~$X$ is a solution, each set in $\hv$ induces a 2-plex.

\begin{figure}[t]
  \centering
    \subfigure[\FISG{}s that will be found in Phase 1.]{
      \includegraphics{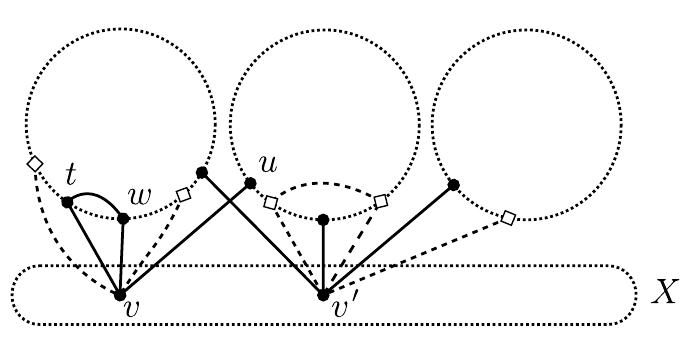}
      \label{fig:pass1}}\hfill
    \subfigure[\FISG{}s that will be found in Phase 2.]{
      \includegraphics{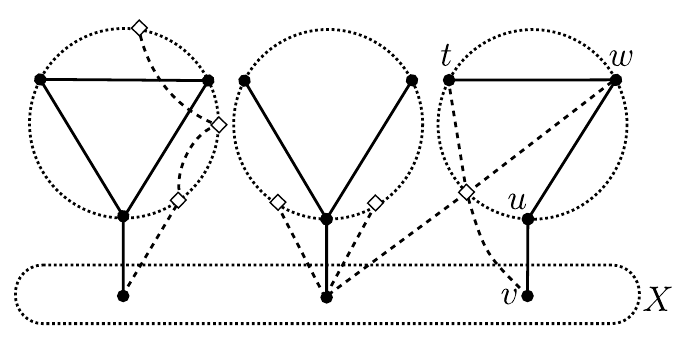}
      \label{fig:pass2}}
    \subfigure[A \FISG{} that will be found in Phase 3.]{
      \includegraphics{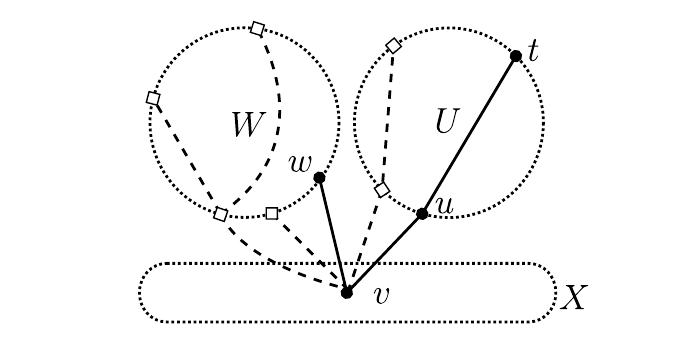}
      \label{fig:pass3}}\hfill
    \subfigure[\FISG{}s that will not be found.]{
      \includegraphics{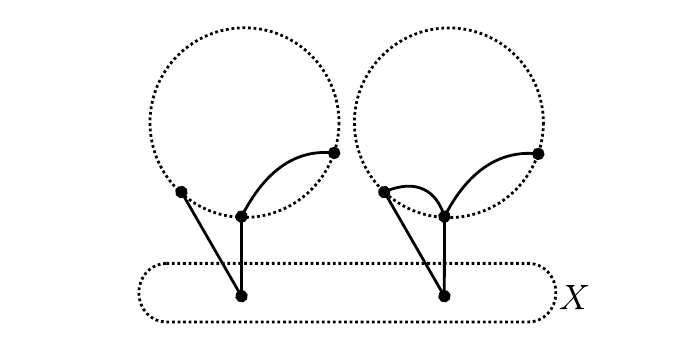}
      \label{fig:no-hit}}
    \caption{Each figure shows the graph~$G$ with a solution~$X$ and \FISG{}s that are found in the different phases of \autoref{mark-neighbors}. Also compare these \FISG{}s with the \FISG{}s shown in \autoref{forbidden}. The vertices~$u,v,w,$ and~$t$ as used in the algorithm are shown. The big circles represent connected components in~$G-X$, that is, they are \pl2es and their vertex sets are sets in~$\hv$. Squares are vertices in the set~$M(v)$ for some vertex~$v\in X$, that is, they are vertices of \FISG{}s that have already been found.}
\end{figure}
We now turn to our first out of three observations. Let~$v\in X$ be a vertex with three neighbors~$u$, $w$, and~$t$. Assume that~$u$ is nonadjacent to~$w$ and~$t$, as shown in \autoref{fig:pass1}. Then, $F:=G[\{t,u,v,w\}]$~is a connected graph, but~$v$ is nonadjacent to \emph{two} vertices~$t$ and~$w$. According to \autoref{fisg-char}, $F$~is a \FISG{}.

\begin{proc}[Phase 1]\label{mark-neighbors}
  Given a graph~$G$ and a solution~$X$, initialize $M(v):=\emptyset$ for each~$v\in X$. For each~$v\in X$, as long as there are vertices~$t,u,w\in N(v)\setminus(M(v)\cup X)$ such that $u$~is neither adjacent to~$t$ nor~$w$, add the vertices~$t,u,$ and~$w$ to~$M(v)$.
\end{proc}\setcounter{proc}{1}

\noindent Now, for each vertex~$v$ in the solution~$X$, let $M(v)$ be the set constructed by Phase~1 of \autoref{mark-neighbors}. For a vertex~$v\in X$, assume that there exists a set~$\hp\in\hv$ such that~$v$ is adjacent to a vertex~$u\in \hp\setminus M(v)$ but nonadjacent to \emph{two} vertices~$t,w\in \hp\setminus M(v)$. This situation is shown in \autoref{fig:pass2}. The graph~$G[\{t,u,w\}]$ is an induced subgraph of~$G[\hp]$. Thus, it is a 2-plex with three vertices, implying that it is connected. Because~$v$~is adjacent to~$u$, the vertex~$v$ is connected but nonadjacent to the \emph{two} vertices~$t$ and~$w$. By \autoref{fisg-char}, $G[\{t,u,v,w\}]$~is a \FISG{}. We continue \autoref{mark-neighbors} as follows:

\begin{proc}[Phase 2]
  \begin{samepage}
    For each~$v\in X$, as long as there is a set~$H\in\hv$ such~that
    \begin{compactenum}
    \item the vertex $v$ is adjacent to a vertex~$u\in H\setminus M(v)$ and
    \item the vertex $v$ is nonadjacent to two vertices~${t,w\in H\setminus M(v)}$,
    \end{compactenum}
    add the vertices~$t,u,$ and~$w$ to~$M(v)$.
  \end{samepage}
\end{proc}\setcounter{proc}{1}  

\noindent Now, for each vertex~$v$ in a solution~$X$, let $M(v)$ be the set constructed by Phase~1 and Phase~2 of \autoref{mark-neighbors}. Assume that for a vertex~$v\in X$, there exist two sets~${U,W\in\hv}$ such that there exist two neighbors~$u\in U\setminus M(v)$ and~$w\in W\setminus M(v)$ of~$v$. This situation is shown in \autoref{fig:pass3}. Assume that~$U\setminus M(v)$ or~$W\setminus M(v)$ contains at least three vertices. Without loss of generality, assume that $|U\setminus M(v)|\geq 3$. Then, $G[U\setminus M(v)]$~is a connected 2-plex. Therefore, there exists a neighbor~$t\in U\setminus M(v)$ of~$u$. The vertex~$w$ is nonadjacent to~$t$ and $u$, because~$w$ is in another set in~$\hv$. Because $F:=G[\{t,u,v,w\}]$~is connected, $F$~is a \FISG{} according to \autoref{fisg-char}.

\begin{proc}[Phase 3]
  \begin{samepage}
  For each vertex~$v\in X$, as long as there are two vertex
  sets~$U,W\in\hv$ such that
  \begin{compactenum}
  \item the vertex~$v$ has neighbors~$u\in U\setminus M(v)$ and $w\in W\setminus M(v)$ and
  \item there is a neighbor~$t\notin X\cup M(v)$ of either~$u$ or~$w$,
  \end{compactenum}
  add the vertices~$t,u,$ and~$w$ to~$M(v)$. Finally, return $M(v)$ for all vertices $v\in X$.
  \end{samepage}
\end{proc}

\noindent This concludes the description of \autoref{mark-neighbors}.  For a solution~$X$, we now inspect the union~$M:=\bigcup_{v\in X}M(v)$ of the sets~$M(v)$ constructed by \autoref{mark-neighbors}. Informally speaking, we show that if we remove~$M$ from~$G$, then each vertex~$v\in X$ is adjacent to the vertices of at most one large connected component in~${G-(X\cup M)}$. As a result, there are at most~$|X|$ large connected components in~$G-(X\cup M)$ containing neighbors of~$X$.  Further, we show that if~$v$ is adjacent to vertices of a connected component in~$G-(X\cup M)$, then it is adjacent to almost all of its vertices. This helps us in \autoref{redundant} to bound the sizes of the connected components in~$G-X$. To formalize these properties, we introduce the concept of a \emph{peripheral set}:

\begin{definition}\label{per} 
  Let~$X$ be a solution. We call a vertex set~$M$ with the following properties \emph{peripheral with respect to $X$}:
  \begin{compactenum}
  \item\label{per1} For each vertex~$v\in X$, there are at most two sets~$\hp\in\hv$ such that~$\hp\setminus M$ is adjacent~to~$v$.
  \item\label{per2} If there is a vertex~$v\in X$ and a set~$\hp\in\hv$ such that~$H\setminus M$ is adjacent to~$v$, then~$v$ is nonadjacent to at most one vertex in~$H\setminus M$.
  \item\label{per3} For each vertex~$v\in X$, if there is more than one set~$H\in\hv$ such that~$H\setminus M$ is adjacent to~$v$, then each such set~$H$ satisfies $|H\setminus M|\leq 2$.
  \end{compactenum}
\end{definition}

\begin{figure}[t]
  \centering
  \includegraphics{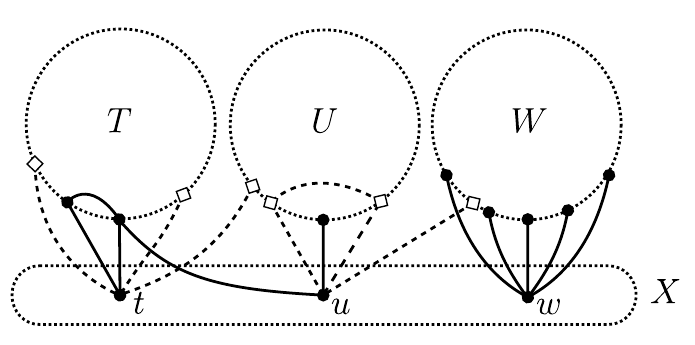}
  \caption{An example for a peripheral set~$M$, which contains the vertices drawn as squares. Shown is the graph~$G$ with a solution~$X$. The circles represent sets in~$\hv$, which induce connected components in~$G-X$.
  \label{per-fig}}
\end{figure}

\noindent For an example, refer to \autoref{per-fig}.  In this figure, no vertex in~$X$ is adjacent to the three sets~$T\setminus M$, $U\setminus M$, and~$W\setminus M$, as required by \autoref{per}(\ref{per1}). The vertex~$u$ is adjacent to~$T\setminus M$ and~$U\setminus M$. There is only one vertex in~$T\setminus M$ that is nonadjacent to~$u$, as required by \autoref{per}(\ref{per2}). As required by \autoref{per}(\ref{per3}), the sets~$T\setminus M$ and~$U\setminus M$ each contain at most two vertices. The vertex~$w$ is only adjacent to~$W\setminus M$. Because~$W\setminus M$ contains more than two vertices, $w$ is only adjacent to~$W\setminus M$, as required by \autoref{per}(\ref{per3}).

\begin{lemma}\label{two-sets}
  \prem{mark-neighbors} The set~$M$ is peripheral with respect to~$X$.
\end{lemma}

\begin{proof}
  We do not directly prove that for each vertex~$v\in X$, the set $M$ satisfies the properties in \autoref{per}. Instead, we show for each vertex~$v\in X$ that the set~$M(v)$ satisfies them. Because~$M(v)\subseteq M$ for all~$v\in X$, this is sufficient. We show the properties separately.

  (\ref{per1}) Assume that there exists a vertex $v\in X$ and three sets~$T,U,W\in\hv$ such that~$v$ has the neighbors~$t\in T\setminus M(v), u\in U\setminus M(v),$ and $w\in W\setminus M(v)$.  This case is illustrated for the vertex~$v'$ in \autoref{fig:pass1}. Because the vertices $t,u,$~and~$w$ come from different connected components in~$G-X$, they are pairwise nonadjacent. Phase~1 of \autoref{mark-neighbors} would have added~$t,u,$~and~$w$ to~$M(v)$.  This contradicts the assumption that~$t\in T\setminus M(v),u\in U\setminus M(v),$ and $w\in T\setminus M(v)$. This shows the first property.

  (\ref{per2}) Assume that there exists a set~$H\in\hv$ such that the vertex~$v\in X$ is adjacent to the vertex~$u\in H\setminus M(v)$ and~$v$ is nonadjacent to the vertices~$t\in H\setminus M(v)$ and~${w\in H\setminus M(v)}$. This is illustrated in \autoref{fig:pass2}. Phase 2 of \autoref{mark-neighbors} would have added the vertices~$t,u,$ and~$w$ to~$M(v)$. This contradicts the assumption that~$t,u,w\in \hp\setminus M(v)$. This shows the second property.

  (\ref{per3}) Assume that there exist two sets~$U,W\in\hv$ such that a vertex~$v\in X$ has the neighbors~$u\in U\setminus M(v)$ and~$w\in W\setminus M(v)$. Without loss of generality, assume that~$|U\setminus M(v)|> 2$.  This situation is shown in \autoref{fig:pass3}.  Because $|U\setminus M(v)|>2$, the 2-plex~$G[U\setminus M(v)]$ is connected. Therefore, the vertex $u$~has a neighbor~$t\in U\setminus M(v)$. Phase 3 of \autoref{mark-neighbors} would have added~$t,u,$~and~$w$ to~$M(v)$. This contradicts the assumption that~$u\in U\setminus M(v)$ and~$w\in W\setminus M(v)$. To fully prove the third property, one can show~$|W\setminus M(v)|\leq 2$ analogously.
\end{proof}

\noindent In the following, we provide a more detailed view on the execution steps of \autoref{mark-neighbors} and also analyze its running time. The following lemma enables us to execute Phase~3 of \autoref{mark-neighbors} quickly.

\begin{lemma}\label{two-neighbors}
  Let~$X$ be a solution. For each vertex~$v\in X$, let~$M(v)$ be the set constructed by Phase 1 of \autoref{mark-neighbors}. If there exists a vertex~$v\in X$ and two sets~$U,W\in\hv$ such that~$U\setminus M(v)$ and~$W\setminus M(v)$ are adjacent to~$v$, then $|N(v)\setminus(M(v)\cup X)|=2$.
\end{lemma}

\begin{proof}
  Assume that the vertex~$v\in X$ has three neighbors~$t,u,w\notin M(v)\cup X$, as shown in \autoref{fig:pass1}. According to the proof of \autoref{two-sets}, there are at most two sets~$U,W\in\hv$ such that~$v$ is adjacent to~$U\setminus M(v)$ and~$W\setminus M(v)$. Without loss of generality, assume that $t,w\in W\setminus M(v)$ and~$u\in U\setminus M(v)$. The vertices~$t,u,$ and $w$ are neighbors of~$v$ and $u$~is neither adjacent to~$t$~nor~$w$. Phase 1 of \autoref{mark-neighbors} would have added~$t,u,$~and~$w$ to~$M(v)$. This contradicts the assumption that~$t,u,w\notin M(v)\cup X$.
\end{proof}

\begin{lemma}\label{mark-neighbors-time}
  Given a solution~$X$, \autoref{mark-neighbors} can be carried out in $O(|X|n^2)$~time.
\end{lemma}

\begin{proof}
  Given a graph~$G$ and a solution~$X$, we first compute the graph~$G-X$ in ${O(n+m)}$~time. We can then compute the collection~$\hv$ of vertex sets of the connected components in~$G-X$. This can be done in $O(n+m)$~time using breadth-first search. During the construction of~$\hv$, we construct a table~$T$ that stores, for each vertex~$u$, the set~$\hp\in\hv$ with $u\in \hp$. We assume that set membership can be tested in constant time and that elements can be added to sets in constant time.  For each vertex~$v\in X$, we now execute the three phases:

  In Phase 1, we construct the set~$N(v)\setminus (M(v)\cup X)$ in $O(n)$~time. For each vertex $u\in N(v)\setminus (M(v)\cup X)$, we scan the set~$N(v)\setminus (M(v)\cup X)$ again to find two vertices nonadjacent to~$u$. Therefore, Phase 1 runs in $O(n^2)$~time for each vertex~$v\in X$.

  In Phase 2, for each~$u\in N(v)$, we can (using the table~$T$) find~$\hp\in\hv$ with~$u\in \hp$ in constant time. If~$u\in X$ or~$u\in M(v)$, then we proceed with the next~$u\in N(v)$. Otherwise, in $O(n)$~time, we scan~$\hp\setminus M(v)$ for two vertices that are nonadjacent to~$v$. The running time for one vertex~$u\in N(v)$ is thus $O(n)$, resulting in a running time of~$O(n^2)$ for each~$v\in X$.

  In Phase 3, we first construct the set $N(v)\setminus(M(v)\cup X)$ in $O(n)$~time. According to \autoref{two-neighbors}, if we have ${|N(v)\setminus(M(v)\cup X)|\ne 2}$, then there is at most one set~$\hp\in\hv$ such that~$\hp\setminus M(v)$ is adjacent to~$v$. Thus, we continue with the next~$v\in X$. Otherwise, let~$u,w\in N(v)\setminus(M(v)\cup X)$. In constant time, we check (using the table~$T$) if the vertices $u$~and~$w$ are in different sets in~$\hv$. If so, we scan the neighborhoods of~$u$ and~$w$ for a vertex~$t\notin X\cup M(v)$ in $O(n)$~time. Thus, the total running time of Phase~3 is~$O(n)$ for each~$v$. \autoref{mark-neighbors} has a worst-case running time of~$O(|X|n^2)$.
\end{proof}

\noindent Note that, given a vertex~$v$ of a solution~$X$, \autoref{mark-neighbors} only finds a \FISG{}~$F$ containing~$v$ if the vertices in~$F-\{v\}$ are neighbors of~$v$ or if at least two vertices of~$F-\{v\}$ are in distinct connected components in~$G-X$. This is not the case for the \FISG{}s shown in \autoref{fig:no-hit}. Thus, \autoref{mark-neighbors} does not necessarily find them. We could search for these \FISG{}s, but this would presumably increase the asymptotic running time of \autoref{mark-neighbors}. It would not improve the worst-case size of our problem kernel.

\subsection{Reducing the Number of Connected
  Components}\label{unmarked-bounds}
In this section, given a solution~$X$ for the graph~$G$, we present data reduction rules to bound the number of connected components in~${G-X}$ by a function only depending on the parameter~$k$. To this end, we bound the size of the peripheral set constructed by \autoref{mark-neighbors} using the following data reduction rule, which is based on \autoref{high-occurence}.

\begin{rul}\label{high-occurence-fast}
  \premv{mark-neighbors} If there exist a vertex~$v\in X$ such that~$|M(v)|>3k$, then delete~$v$ from~$G$ and~$X$ and decrement~$k$ by one.
\end{rul}

\begin{lemma}\label{high-occurence-time}
  \autoref{high-occurence-fast} is correct. Given a solution~$X$ and the set~$M(v)$ constructed by \autoref{mark-neighbors} for each vertex~$v\in X$, we can exhaustively apply \autoref{high-occurence-fast} in $O(|X|n+m)$~time.
\end{lemma}

\begin{proof}
  If \autoref{mark-neighbors} adds vertices to~$M(v)$ for a vertex~$v\in X$, then it has found a \FISG{} that contains no vertices from~$M(v)$. That is, apart from~$v$, this \FISG{} does not contain vertices from previously found \FISG{}s. Thus, if $|M(v)|>3k$, then $M(v)$~contains vertices of more than $k$~\FISG{}s that pairwisely intersect only in the vertex~$v$. According to \autoref{high-occurence}, we can delete~$v$ from~$G$ and decrement the parameter~$k$~by one. For each vertex~$v\in X$, the elements in~$M(v)$ can be counted in $O(n)$~time. The deletion of all vertices~$v\in X$ with~$|M(v)|>3k$ is possible in $O(n+m)$~time.
\end{proof}

\noindent Observe that for each vertex~$v$ in a solution $X$, \autoref{high-occurence-fast} does not change the set~$M(v)$ constructed by \autoref{mark-neighbors}. Also, the graph~$G-X$ is invariant under \autoref{high-occurence-fast}; so is the set~$\hv$. We can conclude that, after we have applied \autoref{high-occurence-fast} to~$G$ and~$X$, the proof of \autoref{two-sets} is still valid and shows that the set $\bigcup_{v\in X}M(v)$ is still peripheral by \autoref{per}. Therefore, \autoref{high-occurence-fast} does not only reduce the size of~$G$ and~$X$; we also obtain a smaller peripheral set. This is because after the exhaustive application of \autoref{high-occurence-fast}, for each vertex~$v\in X$, the set~$M(v)$ contains at most~$3k$ vertices.

\begin{korollar}\label{size-M}
  Let~$X$ be a solution for~$G$. For each vertex~$v\in X$, let~$M(v)$ be the set constructed by \autoref{mark-neighbors}. After exhaustively applying \autoref{high-occurence-fast} to~$G$ and~$X$, the peripheral set $M:=\bigcup_{v\in X} M(v)$ contains at most $3k|X|$ vertices.
\end{korollar}

\noindent 
Now that we have bounded the size of the peripheral set, we can, given a solution~$X$, bound the number of connected components in~${G-X}$. First, we remove connected components from~$G-X$, which are induced by the vertex sets in~$\hv$, according to the following data reduction rule. Then, we use a peripheral set to show a bound on the number of the remaining connected components.

\begin{rul}\label{isolated}
  Let~$X$ be a solution. If there exists a set~$\hp\in\hv$ that is nonadjacent to~$X$, then remove the vertices in~$\hp$ from~$G$.
\end{rul}

\begin{lemma}\label{isolated-time}
  \autoref{isolated} is correct. Given a solution~$X$, we can exhaustively apply \autoref{isolated} in $O(n+m)$~time.
\end{lemma}

\begin{proof}
  Let~$\hp\in\hv$ be the set of vertices chosen for removal by \autoref{isolated} and let~${G':=G-\hp}$. To prove the correctness of \autoref{isolated}, we have to show that $(G',k)$~is a yes-instance if and only if $(G,k)$~is a yes-instance. If $(G,k)$~is a yes-instance, then there exists a solution~$S$ with~$|S|\leq k$ for~$G$. Since $G-S$~is a \pcg 2, $G'-S$~is a \pcg 2 as well. Thus, $(G',k)$~is a yes-instance.

  If $(G',k)$~is a yes-instance, then there exists a solution~$S$ with~$|S|\leq k$ for~$G'$. Because \autoref{isolated} chooses to remove the vertices in~$H$ from~$G$, the set~$\hp$ is nonadjacent to the solution $X$. Therefore, $\hp$~induces an isolated 2-plex in~$G$. It can therefore not contain vertices of a \FISG{}. Thus, also $G-S$~is a 2-plex cluster graph and $(G,k)$~is a yes-instance.

  Considering the running time, we can obtain the set~$\hv$ in $O(n+m)$~time. During the construction of~$\hv$, we use a table~$T$ to store for each vertex~$u$ the set~${\hp\in\hv}$ with~$u\in \hp$. We have already used this technique in the proof of \autoref{mark-neighbors-time}. We construct a further table~$T'$ as follows: for each vertex~$v\in X$ and for each vertex~$u\in N(v)\setminus X$, we set~${T'[T[u]]=1}$. This can be done in~$O(n+m)$ time. Then, the sets~${\hp\in\hv}$ with~$T'[\hp]=0$ are known to have no neighbor in~$X$. These can be removed from~$G$ in~$O(n+m)$ time.
\end{proof}

\begin{figure}
  \centering

  \subfigure[A set~$H\in\hv$ with~$H\setminus M$ nonadjacent to~$X$]{ \hspace{1cm}\includegraphics{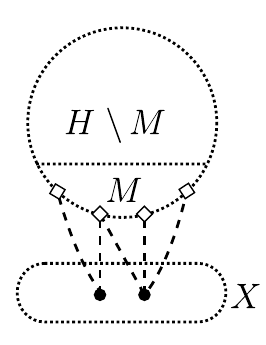}
    \label{fig:h0}\hspace{1cm}}\hspace{1cm}
  \subfigure[A set~$H\in\hv$ with~$H\setminus M$ adjacent to~$X$]{ \hspace{1cm} \includegraphics{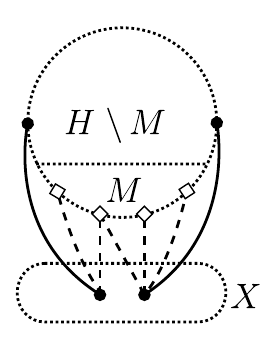}\hspace{1cm}
    \label{fig:h1}}
  \caption{A solution~$X$ and a vertex set~$M$. The big circles represent sets in~$\hv$, or connected components in~$G-X$, likewise.}
  \label{teilung-h1h0}
\end{figure}
\noindent 
Given a solution~$X$ and a vertex set~$M$, there are two possible scenarios for a connected component in~$G-X$. Consider the vertex set~$H\in\hv$ of such a connected component. As shown in \autoref{fig:h0}, it might be the case that the edges between the set~${H\cap M}$ and the solution~$X$ separate the vertices in~$H$ from the vertices in~$X$. That is, the set~$H\setminus M$ might be nonadjacent to~$X$. As shown in \autoref{fig:h1}, it might also be the case that for a set~$H\in\hv$, the set~$H\setminus M$ is adjacent to~$X$. According to \autoref{per}(\ref{per1}), if~$M$ is peripheral, then there are at most~$2|X|$ sets~$H\in\hv$ such that~$H\setminus M$ is adjacent to~$X$. To bound the total number of connected components in~$G-X$, it is left to bound the number of sets~$H\in\hv$ such that~$H\setminus M$ is nonadjacent to~$X$.

\begin{lemma}\label{sets-in-hm}
  Let~$X$ be a solution and let~$M$ be a vertex set. After applying \autoref{isolated}, there are at most $|M|$~sets~$\hp\in\hv$ such that~$\hp\setminus M$ is nonadjacent to~$X$.
\end{lemma}

\begin{proof}
  Let~$\hp\in\hv$ such that~$\hp\setminus M$ is nonadjacent to the solution~$X$. Because \autoref{isolated} has been applied, the set $\hp$ must be adjacent to~$X$. Otherwise, \autoref{isolated} would have removed~$\hp$. Because the set~$\hp\setminus M$ is nonadjacent to~$X$, the set $\hp$ must contain a vertex from~$M$ that is adjacent to~$X$. Because a vertex in~$M$ can be contained in only one set in~$\hv$, there can be at most~$|M|$~sets $H\in\hv$ such that~$H\setminus M$ is nonadjacent to~$X$.
\end{proof}

\noindent Given a solution~$X$ and a peripheral set~$M$, we conclude from \autoref{per}(\ref{per1}) and \autoref{sets-in-hm} that the number of the connected components in~$G-X$ is at most~$2|X|+|M|$.

\section{Bounding the Sizes of Connected Components}
\label{redundant}
In this section, given a solution~$X$ for~$G$, we bound the \emph{sizes} of the connected components in~$G-X$ by functions only depending on the parameter~$k$. Because we have already bounded the size of~$X$ and the \emph{number} of connected components in~$G-X$, this will finally lead to a problem kernel, as we have discussed in the beginning of \autoref{2pvd}. In \autoref{strong}, we present a generalization of the \emph{module} concept\cite{Gal67,DBLP:journals/dm/McConnellS99}. Based on this, we develop a data reduction rule to reduce the sizes of the connected components in~$G-X$. \autoref{ddot-mu} deals with the efficient execution of this data reduction rule and uses a peripheral set~$M$ to bound the sizes of the connected components.  In \autoref{weak}, we present an additional data reduction rule that is only applicable to connected components induced by sets~$H\in\hv$ such that~$H\setminus M$ is nonadjacent to~$X$. We have already specially handled this type of connected components in \autoref{unmarked-bounds}, where we bounded the \emph{number} of connected components in~$G-X$. We use the fact that the edges between the set~$H\cap M$ and the solution~$X$ separate the vertices in~$H$ from the vertices in~$X$, as shown in \autoref{fig:h0}. We will see that the additional data reduction rule presented in \autoref{weak} is necessary to obtain an~$O(k^2)$-vertex problem kernel.

\subsection{Data Reduction Based on Modules}
\label{strong}
Given a solution~$X$, we now develop a characterization of vertices that can be removed from the connected components in~$G-X$. This characterization is based on so-called \emph{modules}\cite{Gal67,DBLP:journals/dm/McConnellS99}. For a graph with the vertex set~$V$, a vertex subset~$Z\subseteq V$ is called a \emph{module}, if any two vertices~$u,v\in Z$ satisfy~$N(v)\setminus Z=N(u)\setminus Z$. That is, a vertex not in~$Z$ is adjacent to either to all or to no vertices in~$Z$. For example, the two vertices~$w$ and~$x$ in \autoref{fig:redundant1} form a module. Modules also serve as the base of the \emph{critical clique} concept introduced by Guo\cite{Guo09} to kernelize the \name{Cluster Editing} problem.

Given a vertex set~$W\subseteq V$, we generalize the module concept and introduce the \emph{\md W}. We call a vertex set~$Z\subseteq V$ a \emph{\md{W}}, if any two vertices~$u,v\in Z$ satisfy~$N(u)\cap W=N(v)\cap W$. That is, a vertex in~$W$ is either adjacent to all or to no vertices in~$Z$. \autoref{fig:redundant} shows examples for~$W$-modules. Observe that if~$Z\subseteq V$ is a $(V\!\setminus\!Z)$-module, then~$Z$ is a module. Every subset of a $W$-module is again a $W$-module.

For a graph~$G$ and a solution~$X$, we use the fact that the vertices in an \md X are equivalent with respect to their neighborhood in~$X$. The idea is, informally, to represent a large \md X by one of its subsets and to replace the \md X by its representative.
\begin{figure}[t]
  \centering
    \subfigure[Graph prior to reduction.]{
     \hspace{0.5cm} \includegraphics{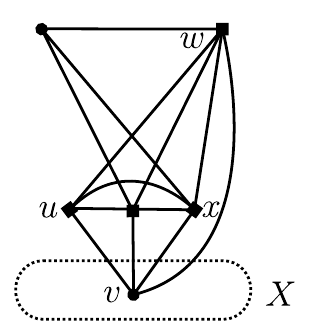}\hspace{0.5cm}
      \label{fig:redundant1}}\hspace{1cm}
    \subfigure[Removed~$w$ and~$x$: valid data reduction.]{
      \hspace{0.5cm}\includegraphics{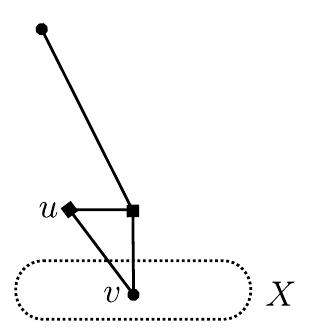}\hspace{0.5cm}
      \label{fig:redundant2}}\hspace{1.0cm}
    \subfigure[Removed~$u$: wrong data reduction.]{
      \includegraphics{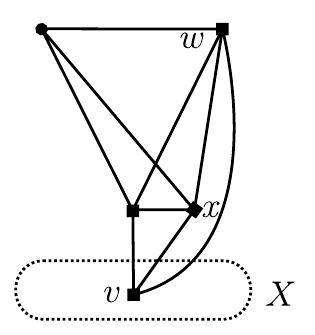}
      \label{fig:redundant3}}
    \caption{In each displayed graph, the vertices drawn as squares
      form an \md X.}
    \label{fig:redundant}
\end{figure}
Consider the following example, which also shows that we cannot choose an arbitrary subset of an \md X as representative: the graph shown in \autoref{fig:redundant1}, call it~$G'$, requires one vertex deletion to transform it into a \pcg 2. The vertices~$u,w,$ and~$x$ are part of an~$X$-module. Observe that also for~$G'-\{w,x\}$ shown in \autoref{fig:redundant2}, one vertex deletion is required to transform it into a \pcg2. It follows that $(G',k)$ is a yes-instance if and only if~$(G'-\{w,x\},k)$~is. Therefore, it is valid to remove~$w$ and~$x$ from~$G'$ to obtain the graph shown in \autoref{fig:redundant2}. In contrast, the graph~$G'-\{u\}$ shown in \autoref{fig:redundant3} is a \pcg2. Because $G'-\{u\}$ can be transformed into a \pcg 2 with less vertex deletions than~$G'$, we may not remove $u$ from~$G'$. To circumvent this problem, we give a constraint on the vertices that may be removed from an \md X in~$G$. Recall that the connected components in~$G-X$ are induced by vertex sets in~$\hv$.

\begin{definition}\label{redundant-def}
  Let~$X$ be a solution. For $\hp\in\hv$, let~$R(H)\subseteq H$ be an \md X. We call~$R(H)$ \emph{redundant} if there exists an \md X $Z(H)$ with $R(H)\subseteq Z(H)\subseteq H$ that contains all vertices from~$H$ that are nonadjacent to a vertex in~$R(H)$.
\end{definition}

\begin{rul}\label{rul:goodconnected}
  Let~$X$ be a solution, let $\hp\in\hv$ and let $R(\hp)$ be a redundant subset of~$H$. If $|R(\hp)|> k+3$, then choose an arbitrary vertex from~$R(\hp)$ and remove it from~$G$.
\end{rul}

\noindent In \autoref{ddot-mu} we construct a redundant set $R(H)$ for each vertex set~${H\in\hv}$ so that we can give a bound on the size of~$\hp\setminus R(\hp)$. Using \autoref{rul:goodconnected}, we can then bound the size of~$R(H)$.  To prove the correctness of the above data reduction rule, we assume that \autoref{rul:goodconnected} chooses to remove a vertex~$u$ from~$G$ and show that~$(G,k)$ is a yes-instance if and only if~$(G-\{u\},k)$ is a yes-instance. To this end, we need three further observations, which we present in the following lemmas.

\begin{lemma}\label{FISG-has-v}
  Let~$\mathcal G$ be an arbitrary graph and let $v$~be a vertex of~$\mathcal G$. If $\mathcal G-\{v\}$ but not~$\cal G$~is a 2-plex cluster graph, then $\cal G$~contains a \FISG{} including the vertex~$v$.
\end{lemma}

\begin{proof}
  Because $\cal G$~is not a 2-plex cluster graph, it contains a \FISG{}.  If all \FISG{}s in~$\cal G$ did not contain~$v$, then no \FISG{} could be destroyed by removing $v$~from~$\cal G$. Thus, $\mathcal G-\{v\}$~would not be a 2-plex cluster graph, contradicting our assumption.
\end{proof}

\noindent Additionally to the assumption that \autoref{rul:goodconnected} chooses to remove a vertex~$u$ from~$G$, we now assume that~$(G-\{u\},k)$ is a yes-instance and show two further lemmas. Finally, we prove the correctness of \autoref{rul:goodconnected}.

\begin{asu}\label{working-asu}
  Let~$X$ be a solution and let $R(\hp)$~be a redundant subset of ${\hp\in\hv}$. Assume that \autoref{rul:goodconnected} chooses to remove the vertex ${u\in R(\hp)}$ from~$G$. Further, assume that~$(G-\{u\},k)$ is a yes-instance, that is, that there exists a solution~$S$ with~${|S|\leq k}$ for the graph~$G-\{u\}$.
\end{asu}

\noindent In the following, we write $G'$~for~$G-\{u\}$. Because we assume that \autoref{rul:goodconnected} chooses to remove~$u$ from~$R(H)$, the set $R(H)$ must contain more than~$k+3$ vertices, which implies $|R(\hp)\setminus(S\cup\{u\})|\geq 3$.
Because~$G[H]$ is a \pl 2, $G[R(\hp)\setminus(S\cup\{u\})]$ is a \pl 2 containing at least three vertices. We can conclude that $G[R(\hp)\setminus(S\cup\{u\})]$~is connected. The graph $G[\hp\setminus(S\cup\{u\})]$ is connected for the same reason.

\begin{lemma}\label{connected}
  Under \autoref{working-asu}, let $G-S$~contain a \FISG{}~$F$ including~$u$. Then in~$G'-S$, the vertices of~$F-\{u\}$ are connected to all vertices in~$\hp\setminus(S\cup\{u\})$.
\end{lemma}
\begin{proof}
  Let $v$~be a vertex of~$F-\{u\}$. Because $F$~is connected, there exists a path in~$G-S$, connecting $v$~to~$u$. This path has to use a neighbor~$w$ of~$u$ (possibly, $v=w$). We now distinguish between the two cases~$w\in H\setminus S$ and~$w\notin H\setminus S$.
\begin{figure}[t]
  \centering\hspace{0.25cm}\subfigure[Case~$w\in \hp\setminus S$. Because ${\hp\setminus (S\cup\{u\})}$ contains three vertices, the vertices in ${\hp\setminus (S\cup\{u\})}$ are connected.]{ \hspace{1.5cm} \includegraphics{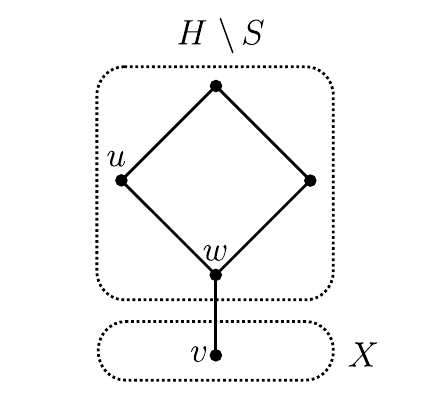}\hfill
    \label{fig:connected-indh}} \hspace{1cm} \subfigure[Case~$w\in X$. Because~$R(H)\setminus S$ is an \md X and~$u$ is adjacent to~$w\in X$, all vertices in~$R(\hp)\setminus S$ are adjacent to~$w$.]{ \hspace{1.25cm} \includegraphics{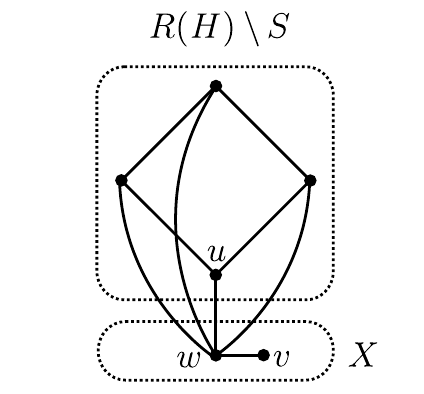}\hspace{1.25cm}
    \label{fig:connected-outdh}}
  \caption{The vertices~$u,v$, and $w$ are named as in the proof of \autoref{connected}. Note that in either case, $v$~is connected to all vertices in~$\hp\setminus(S\cup\{u\})$ even if $u$~is removed. Also note that the vertex $v$~is not necessarily~in~$X$.}
\end{figure}

According to \autoref{working-asu}, $G[\hp\setminus(S\cup\{u\})]$~is connected.  So if $w\in \hp\setminus S$, as shown in \autoref{fig:connected-indh}, then $w$~is connected to every other vertex in~$\hp\setminus(S\cup\{u\})$. That is, $w$ connects~$v$ to the vertices in~$\hp\setminus(S\cup \{u\})$ even when~$u$ is removed.

Because $w$~is in~$G-S$, we have $w\notin S$. That is, if $w\notin \hp\setminus S$, then $w\notin \hp$. Because~$w$ is adjacent to~$u\in R(H)$ and because there are no edges between distinct sets in~$\hv$, we have~$w\in X$, as shown in \autoref{fig:connected-outdh}. Because $u$ is the neighbor of~$w\in X$ and~$u$ is in the \md X~$R(H)$, it follows that all vertices in~$R(\hp)\setminus(S\cup \{u\})$ are neighbors of~$w$ in~$G'-S$. So $w$ connects~$v$ to the vertices in~$\hp\setminus(S\cup\{u\})$ even when~$u$ is removed.
\end{proof}

\begin{lemma}\label{in-hi}
  Under \autoref{working-asu}, let $Z(\hp)$ be an \md X with ${R(H)\subseteq Z(H)\subseteq H}$ and let~$F$ be a \FISG{} in~$G-S$ including~$u$. If a vertex~$v$ of~$F$ is nonadjacent to a vertex~$w\in Z(H)\setminus S$, then~$v\in \hp\setminus S$.
\end{lemma}

\begin{proof}
\begin{figure}
  \centering
  \includegraphics{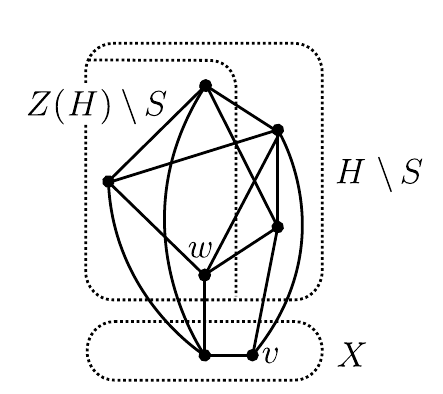}
  \caption{Because $v\in X$~is nonadjacent to the vertex~$w$ of the \md X $Z(H)\setminus S$, the vertex $v$~can not be adjacent to any vertex in~$Z(H)\setminus S$. These are more than three vertices. But $v$ is connected to all vertices in~$\hp\setminus S$, including~$Z(H)\setminus S$.}
  \label{fig:in-hi}
\end{figure}
Assume that a vertex~$v\notin \hp\setminus S$ of~$F$ is nonadjacent to the vertex~${w\in Z(H)\setminus S}$. This situation is shown in \autoref{fig:in-hi}. Because~$v$ is in~$G-S$, we have $v\notin S$ and therefore~${v\notin \hp}$. We first show that~$v$ is nonadjacent to the \md X~$Z(H)\setminus S$.

Assume that~$v$ is adjacent to the \md X~$Z(\hp)\setminus S$. This implies~$v\in X$, because there are no edges between distinct sets in $\hv$ and~$v\notin H$. Because~$w$ is in the \md X~$Z(H)\setminus S$ and because~$v\in X$ is adjacent to $Z(H)\setminus S$, the vertex $v$ must also be adjacent to~$w$. This is by our assumption not the case, so $v$ is nonadjacent to the \md X~${Z(\hp)\setminus S}$. In particular, $v$ is nonadjacent to its subset~$R(\hp)\setminus (S\cup\{u\})$.

According to \autoref{connected}, the vertex $v$~is connected to all vertices in~$R(\hp)\setminus(S\cup\{u\})$ in~$G'-S$.  By \autoref{working-asu}, there are at least three vertices in~$R(\hp)\setminus(S\cup\{u\})$. These are connected but nonadjacent to~$v$ in~$G'-S$. By \autoref{fisg-char}, this implies that there exists a \FISG{} in~$G'-S$, contradicting \autoref{working-asu}.
\end{proof}

\begin{lemma}\label{goodconnected-correct}
  \autoref{rul:goodconnected} is correct.
\end{lemma}

\begin{proof}
  Assume that \autoref{rul:goodconnected} chooses to remove a vertex~$u$ from~$G$. Let $G'$ denote the graph $G-\{u\}$. We have to show that~$(G,k)$ is a yes-instance if and only if~$(G',k)$ is a yes-instance. If $(G,k)$ is a yes-instance, then there exists a solution~$S$ with~$|S|\leq k$ such that $G-S$~is a 2-plex cluster graph. Then, also $G'-S$~is a 2-plex cluster graph and $(G',k)$~is a yes-instance.

  If $(G',k)$~is a yes-instance, then there exists a solution~$S$ with~$|S|\leq k$ such that $G'-S$~is a 2-plex cluster graph, implying that \autoref{working-asu} is true. Assume that $G-S$~contains a \FISG{}. By \autoref{FISG-has-v}, there exists a \FISG{}~$F$ in~$G-S$ containing the vertex~$u$.  Because~$F$ is a \FISG{}, it contains a vertex~$v$ that is connected but nonadjacent to two vertices $w,x$ in~$F$.

  If~$u\notin\{v,w,x\}$, then \autoref{connected} shows that the vertices~$v,w,x$ are connected to all vertices in~$\hp\setminus(S\cup\{u\})$ in~$G'-S$. Thus, the vertices~$v,w,$ and~$x$ would exist in~$G'-S$ and would be connected. That contradicts $G'-S$~being a 2-plex cluster graph, because~$v$ is nonadjacent but connected to the vertices~$w$ and~$x$. Thus, $u$~must be one of~$v,w$ or~$x$.

  First, assume that $u=v$. That is, the vertex $u\in R(\hp)$ is nonadjacent to the vertices~$w$ and~$x$. From \autoref{in-hi}, we can conclude that $w,x\in \hp\setminus S$. Because also~$u\in \hp\setminus S$, this contradicts the graph $G[\hp\setminus S]$~being a 2-plex. So~$u$ must either be~$w$ or~$x$.

  Without loss of generality, assume that~$u=w$. That is, the vertex $u\in R(H)$ is nonadjacent to~$v$. By \autoref{in-hi}, we have $v\in\hp\setminus S$. By \autoref{redundant-def}, there exists an \md X $Z(H)$ with~$R(H)\subseteq Z(H)\subseteq H$ and $v\in Z(H)$, because the vertex~$v\in H\setminus S$ is nonadjacent to the vertex~$u\in R(H)$. But then, because the vertex~$v\in Z(H)$ is nonadjacent to~$x$, the vertex $x$~must also be in~$\hp\setminus S$ by \autoref{in-hi}. This again contradicts~$G[\hp\setminus S]$ being a 2-plex. We conclude that $G-S$~must be a 2-plex cluster graph. Thus, $(G',k)$~is a yes-instance.
\end{proof}

\subsection{Constructing Redundant Sets}
\label{ddot-mu}
In this section, we show how to efficiently find redundant sets as defined in \autoref{redundant-def}. Our goal is, given a solution~$X$ and the vertex set~$H\in\hv$ of a connected component in~$G-X$, to construct a redundant subset~$R(H)\subseteq H$ so that the size of~$H\setminus R(H)$ is bounded by a function only depending on the parameter~$k$. Then, we can apply \autoref{rul:goodconnected} to~$R(H)$ to bound the overall size of~$H$.

To this end, we employ a peripheral set~$M$. Using \autoref{size-M}, we can bound the size of~$M$ by $3k|X|$. Thus, for each~$H\in\hv$, we only need to bound the size of the set~$H\setminus M$. \autoref{per}(\ref{per2}) for peripheral sets guarantees that if a vertex~$v\in X$ is adjacent to~$H\setminus M$, then there is at most one vertex in~$H\setminus M$ that is nonadjacent to the vertex~$v$. Thus, the number of vertices in~$H\setminus M$ that are nonadjacent to a vertex in~$\nx{H\setminus M}$ cannot exceed~$|X|$. The size of~$X$ is in turn bounded by~$4k$ in \autoref{X-boundary}. It follows that we only have to bound the number of vertices in~$H\setminus M$ that are adjacent to \emph{all} vertices in~$\nx{H\setminus M}$. We show that we can obtain a redundant set from such vertices by employing the following~algorithm:

\begin{proc}\label{construct-D}
  Given a set~$M$ that is peripheral with respect to a solution~$X$, for each~$\hp\in\hv$, first find all vertices belonging to~$\hp\cap M$ and $\nx{\hp\setminus M}$. Then, construct the sets
  \begin{align*}
    A(\hp)&:=\{u\in \hp\mid\exists w\in \hp\cap M:u\text{ is nonadjacent to }w\}\text{,}\\
    B(\hp)&:=\{u\in \hp\mid\exists w\in \nx{\hp\setminus M}:u\text{ is nonadjacent to }w\}\text{, and}\\
    C(\hp)&:=\{u\in \hp\mid\exists w\in B(\hp):u\text{ is nonadjacent to }w\}\text{.}
  \end{align*}
  Return~$R(\hp):=\hp\setminus\bar R(\hp)$, where $\bar R(\hp):=A(\hp)\cup B(\hp)\cup C(\hp)\cup(\hp\cap M)$.
\end{proc}

\begin{lemma}\label{construct-D-correct}
  Given a set~$M$ that is peripheral with respect to a solution~$X$, for~$\hp\in\hv$, let~$R(\hp)$ be the set constructed by \autoref{construct-D}. The set~$R(H)$ is redundant.
\end{lemma}
 
\begin{proof}
  According to \autoref{redundant-def}, we have to show that there exists an \md X $Z(H)$ with $R(H)\subseteq Z(H)\subseteq H$ that contains all vertices in~$H$ that are nonadjacent to a vertex in~$R(H)$. Because~$G[H]$ is a \pl2, we could choose~$Z(H):=R(H)$. But with \pl ses in mind, we present a proof that does not rely on the fact that~$G[H]$~is~a~\pl2.

Consider the set $Z(H):=\{u\in H\setminus M\mid \nx{u} = \nx{\hp\setminus M}\}$. For any two vertices $u,v\in Z(H)$, we have that $\nx{u}=\nx{H\setminus M}=\nx v$. Thus, the set~$Z(H)\subseteq H$ is an \md X.  To show that a vertex~$u$ is in~$Z(H)$, it is sufficient to show~$u\in H\setminus M$ and $\nx{\hp\setminus M}\subseteq \nx{u}$. The opposite inclusion $\nx{H\setminus M}\supseteq \nx{u}$ follows directly from~$u\in H\setminus M$.

  We first show that~$R(H)\subseteq Z(H)$. Because~$R(\hp)\cap M=\emptyset$, every vertex in~${R(\hp)}$ is in~$H\setminus M$. Because~$R(\hp)\cap B(\hp)=\emptyset$, for a vertex~$w\in \nx{\hp\setminus M}$, each vertex~${u\in R(\hp)}$ is adjacent to~$w$. Otherwise, $u$~would be in~$B(\hp)$. From this, we can conclude that~${\nx{H\setminus M}\subseteq \nx{u}}$. This implies~$u\in Z(H)$.
 
  Now assume that there exists a vertex $u\in R(\hp)$~and a vertex~$w\in \hp$ such that~$u$ and~$w$ are nonadjacent. From $R(\hp)\cap A(\hp)=\emptyset$ follows that $w\notin M$. Otherwise, $u$~would be in~$A(\hp)$. Because $R(\hp)\cap C(\hp)=\emptyset$, for a vertex~$v\in \nx{\hp\setminus M}$, the vertex $w$~is adjacent to~$v$. Otherwise, $w\in B(\hp)$~and therefore~$u\in C(\hp)$. Thus, we have $\nx{\hp\setminus M}\subseteq \nx{w}$ and~$w\in Z(H)$.
\end{proof}

\begin{lemma}\label{construct-D-time}
  Given a set~$M$ that is peripheral with respect to a solution~$X$, \autoref{construct-D} can be carried out in $O(n^2)$~time.
\end{lemma}

\begin{proof}
  Observe that we can construct the set~$\hv$ in $O(n+m)$~time. During the construction of~$\hv$, we use a table~$T$ to store for each vertex~$u$ the set~$\hp\in\hv$ with~$u\in \hp$.  We now scan each~$\hp\in\hv$ in four passes, classifying each vertex~$u\in H$ as follows:

  The first pass constructs the sets~$\hp\cap M$ and~$\nx{\hp\setminus M}$. If $u\in M$, we memorize the vertex~$u$ to belong to~$\hp\cap M$. If $u\notin M$, we memorize its neighbors in~$X$ to belong to~$\nx{\hp\setminus M}$. Finding~$u$'s neighbors in~$X$ can take~$O(|X|)$ time.

  The second pass constructs the sets~$A(\hp)$ and~$B(\hp)$ with the results from the first pass as follows: if the vertex~$u$ is nonadjacent to a vertex in~$\hp\cap M$, then add~$u$ to~$A(\hp)$. This works in $O(|\hp\cap M|)$~time.
If the vertex~$u$ is nonadjacent to a vertex in~$\nx{\hp\setminus M}$, which can be checked in $O(|X|)$~time, then add $u$ to~$B(\hp)$. 

  The third pass is similar to the second pass and constructs $C(\hp)$ from~$B(\hp)$ in $O(|B(\hp)|)$~time. In a final pass, we add all vertices~$u$ that are not in~$A(\hp),B(\hp),C(\hp)$ or~$M$ to~$R(\hp)$. This can be done in constant time for each vertex~$u$.

  Finally, we encounter at most~$n$ vertices scanning through each~$\hp\in\hv$, yielding a total running time of~$O(n^2)$.
\end{proof}

\begin{lemma}\label{goodconnected-time}
  Given a set~$M$ that is peripheral with respect to a solution~$X$, we can exhaustively apply \autoref{rul:goodconnected} in $O(n^2)$~time.
\end{lemma}
\begin{proof}
  We first, for all~$\hp\in\hv$, use \autoref{construct-D} on the sets~$X$ and~$M$ to construct the sets~$R(\hp)$ in $O(n^2)$~time (\autoref{construct-D-time}). According to \autoref{construct-D-correct}, these sets are redundant. Thus, \autoref{rul:goodconnected} can be applied.

  Observe that after \autoref{rul:goodconnected} removes a vertex~$u\in R(H)$ from~$G$, the set $R(H)\setminus\{u\}$ is still redundant. Thus, we can remove a whole subset of~$R(H)$ from~$G$ without constructing new redundant sets between vertex deletions.

  For each~$\hp\in\hv$, we can count the number of vertices in $R(\hp)$ in $O(|R(\hp)|$)~time. Removing a set of vertices works in $O(n+m)$~time.
\end{proof}

\noindent Given an instance~$(G,k)$ and a solution~$X$ for~$G$, we can now bound the sizes of the connected components in~$G-X$ by a function that only depends on the parameter~$k$.

\begin{lemma}\label{h1-vertices}
  Let the set~$M$ be peripheral with respect to a solution~$X$.  For a set ${\hp\in\hv}$, let~$R(H)$ be the redundant subset constructed by \autoref{construct-D}. After exhaustively applying \autoref{rul:goodconnected} using~$R(H)$, the number of vertices in~$\hp\setminus M$ is at most $|\hp\cap M|+2|\nx{\hp\setminus M}|+k+3$.
\end{lemma}

\begin{proof}
  To prove the above lemma, we study the sets constructed in \autoref{construct-D}. By construction of~$R(\hp)$, we have~$R(\hp)=\hp\setminus\bar R(\hp)$. Observe that because~${\bar R(\hp)\subseteq \hp}$, we also have ${\hp\setminus R(\hp)=\bar R(\hp)}$. Because $G[\hp]$~is a 2-plex, there exists at most one vertex~${u\in \hp}$ for every vertex~$w\in \hp\cap M$ such that $u$~and~$w$ are nonadjacent. Thus, we have ${|A(\hp)|\leq|\hp\cap M|}$. Because~$M$ is peripheral, we can conclude from \autoref{per}(\ref{per2}) that for each vertex~$w\in \nx{\hp\setminus M}$, there is at most one vertex~$u\in \hp\setminus M$ such that~$u$ and~$w$ are nonadjacent. If~$\nx{\hp\setminus M}=\emptyset$, then~$B(\hp)=\emptyset$.  Thus, we have $|B(\hp)|\leq |\nx{\hp\setminus M}|$. Now, again because~$G[\hp]$ is a \pl 2, there exists at most one vertex~$u\in \hp$ for every vertex~$w\in B(\hp)$ such that~$u$ and~$w$ are nonadjacent.  Thus, we have $|C(\hp)|\leq |B(\hp)| \leq |\nx{\hp\setminus M}|$.  This shows that the number of vertices in~$\hp\setminus (R(\hp)\cup M)$ cannot exceed $|\hp\cap M|+2|\nx{\hp\setminus M}|$. To get the total number of vertices in~$\hp\setminus M$, we must add~$|R(H)|$. \autoref{rul:goodconnected} bounds~$|R(\hp)|$ to~$k+3$.
\end{proof}

\subsection{Data Reduction Based on Separators}
\label{weak}
In the previous section, we have, given a solution~$X$, bounded the sizes of the connected components in~$G-X$. Given a peripheral set~$M$, we now present an additional data reduction rule to further reduce the sizes of the connected components induced by vertex sets from the collection $\hneg := \{H\in\hv\mid H\setminus M\text{ is nonadjacent to }X\}$. The vertices in a set~$H\in\hneg$ are \emph{separated} from the vertices in the solution~$X$ by the edges between~$M$ and~$X$, as shown in \autoref{fig:h0}. \autoref{fig:h1} shows an example for a vertex set that is \emph{not} in $\hneg$.  The following observation makes clear why an additional data reduction rule for sets in $\hneg$~is~necessary.

According to \autoref{size-M}, if~$k$ is our parameter, we can employ \autoref{high-occurence-fast} to obtain a peripheral set~$M$ containing at most~$3k|X|$ vertices. By \autoref{sets-in-hm}, exhaustively applying \autoref{isolated} gives us a bound of $|M|$ on the number of sets in $\hneg$. Since we have $|M|\leq3k|X|$, if we bound the size of each set in~$\hneg$ by a function linear in~$k$, then the total number of vertices in sets in $\hneg$ is~$O(|X|k^2)$. To conclude an~$O(|X|k)$-vertex problem kernel, we have to provide a data reduction rule additionally to \autoref{rul:goodconnected}.

For each connected component in~$G-X$ that is induced by a set~$H\in\hneg$, we now bound~$|\hp|$ by a function linear in~$|\hp\cap M|$. Thus, we effectively bound the total number of vertices in sets in $\hneg$ by~$O(|M|)$.  Observe that since~$X$ is a solution, every \FISG{} that contains a vertex from a set~${\hp\in\hneg}$ must also contain a vertex from~$X$. Because $\hp\setminus M$ is nonadjacent to~$X$, the \FISG{} $F$~must also contain a vertex from~$\hp\cap M$. The following data reduction rule is based on the idea that if~$|H\setminus M|$ is too large and contains vertices of \FISG{}s, then we can find a small solution containing the vertices in~$H\cap M$.

\begin{rul}\label{ruecken}
  Let~$X$ be a solution and let~$H\in\hv$. Given a vertex set~$M$
  such that~$H\setminus M$ is nonadjacent
  to~$X$, if $|\hp\setminus M|>|\hp\cap M|+1$, then choose
  a vertex from~$\hp\setminus M$ and remove it from~$G$.
\end{rul}

\noindent To prove the correctness of this data reduction rule, we need a series of observations. To this end, we use the following definition:

\begin{definition}
  For two vertex sets~$U$ and~$W$, we introduce the set~$E(U,W)$ of edges between~$U$ and~$W$. That is, $E(U,W)=\{\{u,w\}\mid u\in U\text{ and }w\in W\text{ are adjacent in }G\}$.  We say that a solution \emph{destroys} an edge~$e$, if the solution contains a vertex incident to~$e$.
\end{definition}

\noindent For a solution~$X$ and the vertex set~$H\in\hv$ of a connected component in~$G-X$, the edges in~$E(H,X)$ \emph{separate} the vertices in~$H$ from the vertices in~$X$. This is shown in \autoref{ruecken-proof}. If a solution~$S$ destroys all edges in~$E(H,X)$, then $G[H\setminus S]$ is an isolated~2-plex.

\begin{lemma}\label{ruecken-lemma}
  Let~$S$ and~$X$ be solutions. Assume that there is a vertex set~$M$
  and a set~$\hp\in\hv$ such that~$H\setminus M$ is nonadjacent
  to~$X$. If~$S$ does not destroy all edges in~$E(\hp,X)$, then it
  contains~$|\hp\setminus M|-1$ vertices from~$\hp\setminus M$.
\end{lemma}
\begin{proof}
  Because the solution $S$ does not destroy all edges in~$E(\hp,X)$,
  there must exist an edge~${e\in E(\hp\setminus S, X\setminus
    S)}$.
\begin{figure}
  \centering
  \includegraphics{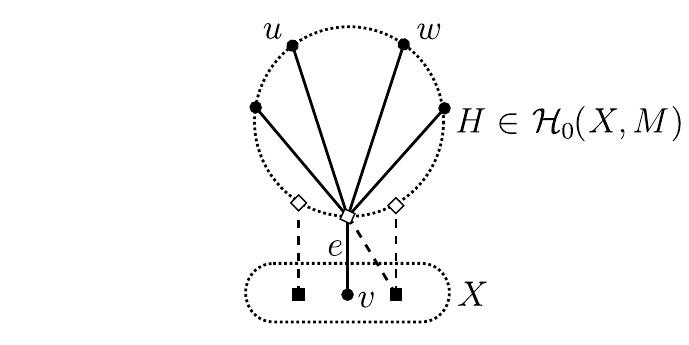}
  \caption{Empty squares are the vertices in the set~$M$. Filled squares are in the solution~$S$. Dashed edges are destroyed by~$S$. Note that there are no edges from $X$ to vertices in~$H\setminus M$. Shown are \FISG{}s that result if a solution~$S$ does not destroy an edge~${e\in E(\hp,X)}$ and if~$S$ does not contain all but one~vertex~in~${\hp\setminus M}$.}
\label{ruecken-proof}
\end{figure}
Now assume that $S$ does not contain two distinct
  vertices~${u,w\in \hp\setminus M}$, as shown in
  \autoref{ruecken-proof}. Because~$u,w\notin M$ and because~$H\setminus M$ is
  nonadjacent to~$X$, the vertex $v\in X\setminus S$ incident to the
  edge~$e$ cannot be adjacent to the vertices~${u,w\in
    \hp\setminus (S\cup M)}$.  But~$\hp\setminus S$ contains at least
  three vertices:~$u,w$ and at least one vertex from~$H\cap M$. Thus,
  the vertex $v$~is connected but nonadjacent to~$u$ and~$w$. We can
  conclude from \autoref{fisg-char} that they are part of a
  \FISG{}. This contradicts $S$ being a solution.
\end{proof}

\begin{lemma}\label{exist-sol}
  Let~$S$ and~$X$ be solutions.  Assume that there is a vertex set~$M$
  and a set~$\hp\in\hv$ such that~$H\setminus M$ is nonadjacent
  to~$X$. If $|\hp\setminus
  M|\geq|\hp\cap M|+1$, then there exists a solution~$S'$
  with~$|S'|\leq|S|$ that destroys all edges in~$E(\hp,X)$.
\end{lemma}

\begin{proof}
  Assume that~$S$ does not destroy all edges in~$E(\hp,X)$. From
  \autoref{ruecken-lemma} and from $|\hp\setminus M|\geq |\hp\cap
  M|+1$, we can conclude that there are at least $|\hp\cap M|$
  vertices from $\hp\setminus M$ in~$S$. The set~$S':=S\cup(\hp\cap M)\setminus(\hp\setminus M)$ destroys all edges in~$E(\hp,X)$. Because~$S$
  contains at least $|\hp\cap M|$ vertices from~$\hp\setminus M$
  and~$S'$ instead contains~$\hp\cap M$, the set $S'$ is not larger
  than~$S$. The set $S'$~is a solution,
  because~$G':=G-(S\cup(\hp\cap M))$ is a 2-plex cluster graph and because
  $G-S'$~is $G'$~with the additional connected component formed
  by the \pl2~$G[\hp\setminus M]$.
\end{proof}

\begin{lemma}\label{ruecken-time}
  \autoref{ruecken} is correct. Given a vertex set~$M$ and a solution~$X$, we can exhaustively apply \autoref{ruecken} in $O(n+m)$~time.
\end{lemma}

\begin{proof}
  Let~$u$ be the vertex chosen by \autoref{ruecken} and let~$G':=G-\{u\}$.
  We have to show that~$(G',k)$ is a yes-instance if and only
  if~$(G,k)$ is a yes-instance. If~$(G,k)$ is a yes-instance, then there
  exists a solution~$S$ with~$|S|\leq k$ for~$G$. Since~$G-S$ is
  a 2-plex cluster graph, $G'-S$ is a \pcg2 as well. Thus, $(G',k)$~is a
  yes-instance.

  If $(G',k)$~is a yes-instance, then there exists a solution~$S$
  with~$|S|\leq k$ for~$G'$. From \autoref{exist-sol}, we can without
  loss of generality assume that the solution~$S$ destroys all edges
  in~$E(\hp,X)$.  Now assume that~$G-S$ is not a 2-plex cluster
  graph. From \autoref{FISG-has-v}, we can conclude that $G$ contains
  a \FISG{}~$F$ including~$u$. The \FISG{}~$F$ also contains a vertex~$v\in
  X\setminus S$, because~$X$ is a solution. However, the vertices~$u$
  and~$v$ are not connected in~$G-S$, because $S$~destroys all edges
  in~$E(\hp,X)$. Therefore, $F$~cannot exist in~$G-S$ and $S$~must be
  a solution for~$G$. Because~$|S|\leq k$, it follows that $(G,k)$~is
  a yes-instance.

  To prove the running time, recall that we can construct the set~$\hv$ in ${O(n+m)}$ time. Then, in~$O(n+m)$ time, we construct a table $T$ so that for every neighbor~$v$ of~$X$, we have~$T[v]=1$.  For each~$H\in\hv$, we now count the number of vertices in~$H\setminus M$ and ${H\cap M}$ in $O(|H|)$~time. If in the counting process, we find a vertex~$v\in H\setminus M$ with~$T[v]=1$, then $H\setminus M$ is adjacent to~$X$. This implies that \autoref{ruecken} is not applicable for~$H$; we continue with the next set in $\hv$. The removal of vertices works in $O(n+m)$~time.
\end{proof}

\begin{korollar}\label{h0-vertices}
  Let~$X$ be a solution. Assume that there is a vertex set~$M$ and a
  set~${\hp\in\hv}$ such that~$H\setminus M$ is nonadjacent to~$X$.
  After exhaustively applying \autoref{ruecken} given~$M$, the set~$\hp\setminus M$ contains at most $|\hp\cap
  M|+1$ vertices.
\end{korollar}

\section{Kernel Size}
\label{result}
In this section, we count the total number of vertices remaining in a graph~$G$ after all data reduction rules have been applied. To this end, we assume that we have a solution~$X$ and a set~$M$ that is peripheral with reference to~$X$. Then, we count the vertices in~$X$, the vertices in~$M$ and the vertices in the connected components in~$G-X$ that are not in~$M$.

Observe that to bound the sizes of the connected components in~$G-X$, which are induced by sets in~$\hv$, we have presented \emph{two} data reduction rules in \autoref{redundant}. \autoref{rul:goodconnected} is applicable to all sets in~$\hv$. The additional \autoref{ruecken} is only applicable to sets in the collection ${\hneg:=\{H\in\hv\mid H\setminus M}$ is nonadjacent~to~$X\}$. Thus, we independently count the vertices in the sets in~$\hneg$ and the vertices in the sets in~$\hpos:=\{H\in\hv\mid H\setminus M$ is adjacent~to~$X\}$. \autoref{fig:h0} shows an example for a set in $\hneg$, \autoref{fig:h1} shows an example for a set in $\hpos$. We have already made this distinction when we bounded the number of sets in~$\hv$ in \autoref{unmarked-bounds}; it is not the only distinction~we~make:

\autoref{per}(\ref{per3}) for peripheral sets ensures that if there is more than one set $H\in\hpos$ such that a vertex~$v\in X$ is adjacent to~$H\setminus M$, then each such set~$H$ satisfies~$|H\setminus M|\leq 2$. To allow for a tighter worst-case analysis, we count the vertices in such sets independently. To this end, we use the following lemma:

\begin{lemma}\label{nx-sum}
Let the set~$M$ be peripheral with respect to a solution~$X$.
For the sets
  \begin{align*}
    X_1 &:=\{v\in X\mid\text{there is exactly one set~$\hp\in\hv$ such that~$\hp\setminus M$ is adjacent to~$v$}\}\text{,}\\
    X_2 &:=\{v\in X\mid\text{there are two or no sets~$\hp\in\hv$ such that~$H\setminus M$ is adjacent to~$v$}\}\\
    &\phantom:=X\setminus X_1\text{ (because~$M$ is peripheral and because of \autoref{per}(\ref{per1})) and}\\
    \mathcal{\tilde H} &:= \{\hp\in\hpos\mid \hp\setminus M \text{ is adjacent to only vertices in }X_1\}\text{,}
  \end{align*}
the following relations hold:
  \begin{align*}
    \sum_{\hp\in\mathcal{\tilde H}}|\nx{\hp\setminus M}|&=|X_1| && \text{and} &
    |\mathcal{\tilde H}| &\leq|X_1| && \text{and} & |\hpos\setminus\mathcal{\tilde
      H}|&\leq 2|X_2|\text{.}
  \end{align*}
\end{lemma}

\begin{proof}
  Let $\hp\in\mathcal{\tilde H}$ be a set such that~$H\setminus M$ is only adjacent to vertices in~$X_1$.  For a vertex~$v\in X_1$ that is adjacent to~$\hp\setminus M$, there is by definition of~$X_1$ no other set~$\hp'\in\hpos$ such that~$\hp'\setminus M$ is adjacent to~$v$. Thus, if we count the number of vertices in~$\nx{\hp\setminus M}$ for all~$\hp\in\mathcal{\tilde H}$, then we count every vertex~$v\in X_1$ exactly once. This proves the first relation.

  For each~$\hp\in \mathcal{\tilde H}\subseteq\hpos$, there is by definition of $\hpos$ at least one vertex~$v\in X_1$ such that~$\hp\setminus M$ is adjacent to~$v$. Thus, $|\mathcal{\tilde H}|\leq|X_1|$.

  According to \autoref{per}(\ref{per1}), there are at most two sets~$\hp\in\hv$ such that~$\hp\setminus M$ is adjacent to~$v$. The set $\hpos\setminus\mathcal{\tilde H}$ only contains sets~${H\in\hpos}$ such that a vertex in~${X_2}$ is adjacent to~$H\setminus M$. This yields $|\hpos\setminus\mathcal{\tilde H}|\leq 2|X_2|$.
\end{proof}

\noindent Given a solution~$X$ and a set~$M$ that is peripheral with respect to~$X$, we now assume that all data reduction rules have been exhaustively applied to our input graph~$G$ and count the vertices in the connected components in~$G-X$ that are not in~$M$.

\begin{lemma}\label{unmarked-boundary}
  Let~$X$ be a solution and let the set~$M$ be peripheral with respect to~$X$. After exhaustively applying \autoref{isolated}, \autoref{rul:goodconnected} and \autoref{ruecken}, it holds that
  \[\big|\smashoperator{\bigcup_{\hp\in\hv}}\,(\hp\setminus M)\big|\leq
  (k+5)|X|+2|M|.\]
\end{lemma}

\begin{proof}Let~$\mathcal{\tilde H}, X_1,$ and~$X_2$ be as
  defined in \autoref{nx-sum}. We can conclude from
  \autoref{h1-vertices} and \autoref{h0-vertices} that
  $|\bigcup_{\hp\in\hv} \hp\setminus
s M|$ is upper-bounded by
  \begin{align*}
    & \smashoperator{\sum_{\hp\in\hpos}}\big(|\hp\cap M|+2\,|\nx{\hp\setminus
      M}|+k+3\big)+\smashoperator{\sum_{\hp\in\hneg}}\big(|\hp\cap
    M|+1\big)\text{.}\\
    \intertext{Because the sets in~$\hv$ are pairwise disjoint, the
      two occurrences of~$|\hp\cap M|$ sum up to a total of~$|M|$,
      yielding}
    & \smashoperator{\sum_{\hp\in\hpos}}\big(2|\nx{\hp\setminus M}|+k+3\big)+|M|+|\hneg|.\\
    \intertext{By \autoref{sets-in-hm}, we have
      that~$|\hneg|\leq|M|$. Thus, the above term is bounded by}
    & \smashoperator{\sum_{\hp\in\hpos}}\big(2|\nx{\hp\setminus M}|+k+3\big) + 2|M|.\\
    \intertext{For each set~$H\in\hpos\setminus\mathcal{\tilde H}$,
      the set~$H\setminus M$ must be adjacent to a vertex
      from~$X_2$. 
This follows from the definition of~$\mathcal{\tilde H}$ in
      \autoref{nx-sum} and by definition of $\hpos$.
 From \autoref{per}(\ref{per3}), we can conclude
      that~$|H\setminus M|\leq 2$, implying that only sets in $\mathcal{\tilde H}$ may actually contain $2|\nx{\hp\setminus M}|+k+3$ vertices that are not in~$M$. We obtain}
    \big|\smashoperator{\bigcup_{H\in\hv}}\,(H\setminus M)\big|&\leq\smashoperator{\sum_{\hp\in\mathcal{\tilde H}}}\big(2|\nx{\hp\setminus
      M}|+k+3\big) + 2|\hpos\setminus\mathcal{\tilde H}| +2|M|.\\
    \intertext{Applying \autoref{nx-sum}, we can bound this by}
    & 2|X_1|+|X_1|(k+3)+4|X_2|+2|M| \leq(5+k)|X_1|+4|X_2| + 2|M|
  \end{align*}
  We can interpret this term as a function in~$|X_1|$
  and~$|X_2|$ with fixed~$|X|$ and~$k\geq 0$. Subject to the
  constraint~$|X_1|+|X_2|=|X|$, it is maximal for~$|X_1|=|X|$ and~$|X_2|=0$. This
  yields the desired result.
\end{proof}

\begin{satz}\label{kern-size}
  \pvd 2 has a problem kernel containing $(10k+6)|X|\leq 40k^2+24k$
  vertices. It can be found in $O(kn^2)$~time.
\end{satz}

\begin{proof}
  Given a \pvd 2 instance~$(G,k)$, we first compute a constant-factor approximate solution~$X$ using \autoref{find-X}. Then, we compute a set that is peripheral with respect to~$X$ using \autoref{mark-neighbors}. We apply \autoref{high-occurence-fast}, from which we obtain a new parameter~$k'\leq k$ and a peripheral set~$M$ with~$|M|\leq 3k|X|$ according to \autoref{size-M}. Finally, we apply \autoref{isolated}, \autoref{rul:goodconnected}, and \autoref{ruecken} to~$G$. The so-obtained graph and the new parameter~$k'$ constitute our problem kernel.

  We first show that after applying all data reduction rules to~$G$, the size of~$G$ only depends on the parameter~$k$. To this end, we count the vertices in the solution~$X$, the vertices in the peripheral set~$M$ and the vertices in~$G-X$ that are not in the peripheral set~$M$. If~$(G,k)$ is a yes-instance, then \autoref{X-boundary} gives an upper bound of~$4k$ on the number of vertices in the constant-factor approximate solution~$X$. If~$X$ is larger, we terminate our kernelization algorithm and output that~$(G,k)$ is a no-instance. By applying \autoref{high-occurence-fast}, we obtain a peripheral set~$M$ that contains at most~$3k|X|$ vertices according to \autoref{size-M}.  By exhaustively applying \autoref{isolated}, \autoref{rul:goodconnected}, and \autoref{ruecken} to~$G$, we can use \autoref{unmarked-boundary} to give a bound of $(k+5)|X|+2|M|=(7k+5)|X|$ on the number of vertices in~$G-X$ that are not in the peripheral set~$M$. Adding~$|X|$ and~$|M|$, we conclude that~$G$ contains at most~${(10k+6)|X|=40k^2+24k}$~vertices.

The correctness of \autoref{high-occurence-fast}, \autoref{isolated}, \autoref{rul:goodconnected}, and \autoref{ruecken}, has been shown in \autoref{high-occurence}, \autoref{isolated-time}, \autoref{goodconnected-correct}, and \autoref{ruecken-time}, respectively.

Finally, we show the running time of our kernelization algorithm. When we construct an approximate solution~$X$ using \autoref{find-X}, we can stop after finding more than~$k$ pairwise vertex-disjoint \FISG{}s, because this implies that~$(G,k)$ is a no-instance. Analog\-ously to the proof of \autoref{find-X-time}, it follows that we can construct~$X$ in $O(k(n+m))$~time. \autoref{mark-neighbors}, \autoref{high-occurence-fast}, \autoref{isolated}, \autoref{rul:goodconnected}, and \autoref{ruecken} run in~$O(kn^2)$ time according to \autoref{mark-neighbors-time}, \autoref{high-occurence-time}, \autoref{isolated-time}, \autoref{goodconnected-time}, and \autoref{ruecken-time}, respectively.
\end{proof}

\noindent To solve a \pvd 2 instance, we can compute a problem kernel with~$O(k^2)$ vertices and reduce this problem kernel to a \hs 4 instance with~$O((k^2)^4)$ sets, as discussed in \autoref{introsec}. Then, we can solve this \hs 4 instance by combining Wahlström's algorithm for \hs 3\cite{Wah07} with iterative compression, as discussed by Dom et al.\cite{DGHNT09}.

\begin{korollar}\label{interleave}
Using \hs 4, we can solve \pvd 2 in ${O(3.076^k+k^8+kn^2)}$ time.
\end{korollar}

\paragraph{Concluding Remarks.} Peripheral sets played a central role in all stages of our kernelization algorithm. After constructing a peripheral set~$M$ with respect to a solution~$X$ using \autoref{mark-neighbors}, the peripheral set~$M$ helps us to bound the number of the connected components in~$G-X$ in \autoref{unmarked-bounds}. For a connected component in~$G-X$, in \autoref{ddot-mu} we use the peripheral set~$M$ to bound the number of vertices that are \emph{not} in the redundant set constructed by \autoref{construct-D}. Then, we remove vertices from that redundant set to bound the overall size of the connected component. In \autoref{weak}, we use the set of edges between~$M$ and~$X$ as a separator to develop an additional data reduction rule to further reduce the sizes of the connected components in~$G-X$.

To construct a set~$M$ that is peripheral with respect to a solution~$X$, we employ \autoref{mark-neighbors}. We could also construct~$M$ by enumerating all minimal \FISG{}s in~$G$, which are shown in \autoref{forbidden}. Then, for each vertex~$v\in X$, we could pick an inclusion-maximal set of \FISG{}s that pairwisely intersect only in~$v$. However, because each minimal \FISG{} contains four vertices, the total number of minimal \FISG{}s in a graph with~$n$ vertices is~$O(n^4)$. In contrast, \autoref{mark-neighbors} finds at most~$O(n)$ \FISG{}s for each vertex~$v\in X$. It runs in~$O(kn^2)$ time. Therefore, the running time of enumerating all minimal \FISG{}s in a graph might be significantly worse that of \autoref{mark-neighbors}.

\chapter[Kernelization for \pvd s]{Kernelization for \boldmath$s$-Plex Cluster Vertex Deletion}
In this chapter, we generalize the problem kernel for \pvd 2 to \pvd s. We will see that many definitions and lemmas that we have worked out for the case~$s=2$ also work for general $s$ if we modify them slightly. In \autoref{per-sec}, we first show how to find an approximate solution~$X$ for a graph~$G$, so that~$G-X$ is an \pcg s. Then, we generalize our concept of a peripheral set and show how to find one. In \autoref{bounds-s}, we revise our data reduction rules to bound the number and the sizes of the connected components in~$G-X$.  In \autoref{result-s}, we conclude a problem kernel with~$O(k^2s^3)$ vertices for \pvd s.

We now turn our attention to the main difference between \pvd 2 and \pvd s.  For \pvd 2, we used the fact that a \pl 2 containing at least three vertices is connected. We used this fact to construct a peripheral set using \autoref{mark-neighbors}, in the correctness proof of \autoref{rul:goodconnected}, and in the correctness proof of \autoref{ruecken}. To generalize these proofs, we need the following result:
 
\begin{lemma}\label{disconnected-size}
  An \pl s containing at least $2s-1$ vertices is a connected graph.
\end{lemma}
 
\begin{proof}
  Let $G=(V,E)$ be an \pl s with more than one connected component. Because $G$ is an \pl s, a vertex in $G$ is nonadjacent to at most $s-1$ other vertices in~$G$.

  Let~$W\subseteq V$ be the vertex set of a connected component of~$G$. Because a vertex in~$W$ is nonadjacent to all vertices in~$V\setminus W$, we have that~$|V\setminus W|\leq s-1$ and~$|W|\leq s-1$. Therefore, it holds that~$|V|\leq 2s-2$. Thus, if an \pl s contains at least~$2s-1$ vertices, it must be a connected graph.
\end{proof}
 
\noindent Note that the bound given in \autoref{disconnected-size} is tight. Consider two cliques with $s-1$ vertices
  each. These two cliques can still be considered as one single~$s$-plex with $2s-2$ vertices.

\section{Approximate Solutions and Peripheral Sets}\label{per-sec}
Given an \pvd s instance $(G,k)$, in this section we first show how to
find an approximate solution~$X$ for~$G$. We then generalize our
concept of peripheral sets and construct a set that is peripheral with
respect to the solution~$X$.

Similarly to the case~$s=2$, we can easily find a constant-factor approximate solution for \pvd s using the algorithm by Guo et al.\cite{DBLP:conf/aaim/GuoKNU09}, which finds an $O(s+\sqrt s)$-vertex \FISG{} in ${O(s(n+m))}$ time if we apply it to a graph that is not a \pcg{s}. In particular, if~$T_s$ is the maximum integer satisfying ${T_s\cdot(T_s+1)\leq s}$, then Guo et al.\cite{DBLP:conf/aaim/GuoKNU09} show that their algorithm finds a \FISG{} with at most~$s+1+T_s$ vertices. Similarly to \autoref{find-X-time}, we can show that \autoref{find-X} computes a constant-factor approximate solution for \pvd s.
 
\begin{lemma}\label{find-X-time-s}
  There is a factor-$(s+1+T_s)$ approximate solution for \pvd s and it can be found in $O(ns(n+m))$~time.
\end{lemma}
\begin{korollar}\label{X-boundary-s}
  Let $(G,k)$ be a yes-instance and let $X$ be a factor-$(s+1+T_s)$ approximate solution for~$G$. Then, $X$ contains $O(sk)$~vertices.
\end{korollar}

\noindent We now construct a set that is peripheral with respect to a
solution~$X$. To this end, we modify \autoref{per}.

\begin{definition}\label{per-s}Let~$X$ be a solution. We call a vertex
  set~$M$ with the following properties \emph{peripheral (with
    respect to $X$)}:
  \begin{compactenum}
  \item\label{per-s1} For each vertex~$v\in X$, there are at most $s$ sets~$\hp\in\hv$ such that~$\hp\setminus M$ is adjacent~to~$v$.
  \item\label{per-s2} If there is a vertex~$v\in X$ and a set~$\hp\in\hv$ such that~$H\setminus M$ is adjacent to~$v$, then~$v$ is nonadjacent to at most $2s-3$ vertices in~$H\setminus M$.
  \item\label{per-s3} For each vertex~$v\in X$, if there is more than one set~$H\in\hv$ such that~$H\setminus M$ is adjacent to~$v$, then each such set~$H$ satisfies~$|\hp\setminus M|\leq 2s-2$.
  \end{compactenum}
\end{definition}

\noindent To construct a peripheral set, we proceed analogously to
\autoref{mark}: for each vertex~$v$ in a given solution~$X$, we find a
\FISG{}~$F$ including~$v$ that contains no vertices from~$M(v)$. Then, we add the vertices of~$F-\{v\}$ to~$M(v)$. We find such
\FISG{}s by three observations, each leading to one of three phases of
an algorithm that constructs the sets~$M(v)$.

We now turn to our first observation. Given a solution~$X$, assume that there exists a vertex~$v\in X$ and a set~${U\subseteq N(v)\setminus X}$ of~$s+1$ neighbors of~$v$ such that~$U$ contains a vertex~$u$ that is nonadjacent to~$U\setminus\{u\}$. Then, the vertex~$u$ is connected to every vertex in~$U$, because the vertices in~$U$ are neighbors of~$v$. The vertex~$u$ is nonadjacent to the~$s$ vertices in~$U\setminus\{u\}$. By \autoref{fisg-char}, the graph~$G[\{v\}\cup U]$ is a \FISG{}.
 
\begin{proc}[Phase 1]\label{mark-neighbors-s}
  \begin{samepage}
    Given a graph~$G$ and a solution~$X$, for each vertex~$v\in X$, let~$M(v)=\emptyset$. For each~$v\in X$, as long as there is a set~${U\subseteq N(v)\setminus (X\cup M(v))}$ such~that
    \begin{compactenum}

    \item $|U|=s+1$ and
    \item there exists a vertex~$u\in U$ that is nonadjacent to~$U\setminus\{u\}$,
    \end{compactenum}
    add the vertices in~$U$ to~$M(v)$.
  \end{samepage}
\end{proc}\setcounter{proc}{0}

\noindent Now, for each vertex~$v$ in the solution~$X$, let $M(v)$ be
the set constructed by Phase~1 of \autoref{mark-neighbors-s}. For a
vertex~$v\in X$, assume that there exists a set~$\hp\in\hv$ such
that~$v$ is adjacent to a vertex~$u\in H\setminus M(v)$. Further,
assume that the vertex~$v$ is nonadjacent to a set~$W\subseteq
H\setminus M(v)$ of $2s-2$ vertices. Then, the graph~$G[\{u\}\cup W]$
is an induced subgraph of the \pl s~$G[\hp]$ and contains $2s-1$
vertices. According to \autoref{disconnected-size}, it is
connected. The vertex~$v$ is, because it is a neighbor of~$u$ and because~$u$ is adjacent to~$W$, connected but nonadjacent to
the~$2s-2$ vertices in~$W$. By \autoref{fisg-char}, the
graph~$G[\{u,v\}\cup W]$ is a \FISG{}. We continue
\autoref{mark-neighbors-s} as follows:
 
\begin{proc}[Phase 2]
  \begin{samepage}
    For each~$v\in X$, as long as there is a set~$\hp\in\hv$ such~that
    \begin{compactenum}
    \item the vertex~$v$ is adjacent to a vertex~$u\in H\setminus M(v)$ and
    \item the vertex~$v$ is nonadjacent to a set~$W\subseteq H\setminus M(v)$ of vertices with~$|W|=2s-2$,
    \end{compactenum}
    add the vertex $u$ and the vertices in~$W$ to~$M(v)$.
  \end{samepage}
\end{proc}\setcounter{proc}{0}

\noindent Now, for each vertex~$v$ in the solution~$X$, let $M(v)$ be
the set constructed by Phase~1 and Phase~2 of
\autoref{mark-neighbors-s}. Assume that for a vertex~$v\in X$, there
are two sets~${U,W\in\hv}$ such that~$v$ is adjacent to the
vertices~$u\in U\setminus M(v)$ and~$w\in W\setminus M(v)$. Further,
assume that~$W\setminus M(v)$ contains at least $2s-1$ vertices. Then,
$G[W\setminus M(v)]$ is a 2-plex containing at least~$2s-1$
vertices. According to \autoref{disconnected-size}, it is
connected. The vertex~$u\in U\setminus M(v)$ is nonadjacent to
at least~$2s-1$ vertices in~$W\setminus M(v)$, but~$F:=G[\{u,v\}\cup W]$ is
connected. According to \autoref{fisg-char}, it is a \FISG{}.
 
\begin{proc}[Phase~3]
  For each~$v\in X$, as long as there are~$U,W\in\hv$~such~that
  \begin{compactenum}
  \item $|W\setminus M(v)|\geq 2s-1$ and
  \item the vertex~$v$ has neighbors~$u\in U\setminus M(v)$ and~$w\in W\setminus
  M(v)$,
  \end{compactenum}
add the vertices~$u$, $w$,
  and~$2s-2$ other vertices from~$W\setminus M(v)$ to~$M(v)$.
\end{proc}
 
\noindent Note that in contrast to \autoref{mark-neighbors}, Phase~2 and Phase~3 of \autoref{mark-neighbors-s} do not necessarily find minimal \FISG{}s. That is, there exist \FISG{}s found by Phase~2 and Phase~3 such that we could remove a vertex from them and they would still be \FISG{}s. For running time considerations, we construct \FISG{}s from parts of \pl{s}es that contain enough vertices to derive their connectedness from \autoref{disconnected-size}. Thus, we do not have to explicitly check whether the subgraphs that we find are connected.
 
\begin{lemma}\label{s-sets}
 \prem{mark-neighbors-s} The set~$M$ is peripheral with respect
  to~$X$.
\end{lemma}

\begin{proof}
  The proof of this lemma is analogous to the proof of
  \autoref{two-sets}. For each vertex~${v\in X}$, the set~$M(v)$ satisfies
  all properties in \autoref{per-s}. This follows directly from the
  description of \autoref{mark-neighbors}.
\end{proof}
 
\begin{lemma}\label{mark-neighbors-time-s}
  Given a solution~$X$, \autoref{mark-neighbors-s} can be carried out in $O(|X|n^2)$~time.
\end{lemma}
 
\begin{proof}
  The running times of Phase~1 and Phase~2 of \autoref{mark-neighbors-s} can be proven in the same way as for \autoref{mark-neighbors} in \autoref{mark-neighbors-time}. We only prove the running time of the modified Phase~3. First, we construct for each vertex~$v\in X$ the set~${N(v)\setminus(M(v)\cup X)}$. The proof of \autoref{mark-neighbors-time} shows how this can be done in~$O(n)$ time. For each vertex~${u\in N(v)\setminus(M(v)\cup X)}$, we can determine the set~$\hp\in\hv$ with~$u\in\hp$ in constant time, as seen in the proof of \autoref{mark-neighbors-time}. Counting the elements in~$\hp\setminus M(v)$ takes at most~$O(n)$ time.  This yields a running time of~$O(|X|n^2)$ for Phase~3 of \autoref{mark-neighbors-s}.
\end{proof}
 
\section{Adapted Data Reduction Rules and Bounds}\label{bounds-s}
\noindent Given an \pvd s instance $(G,k)$ and a solution~$X$ for~$G$,
we now bound the number and the sizes of the connected
components in~$G-X$. To this end, we first revise
\autoref{high-occurence-fast} as shown below.

\begin{rul}\label{high-occurence-fast-s}
  \premv{mark-neighbors-s} If there exists a vertex~$v\in X$ such that~$|M(v)|>2sk$, then delete~$v$ from~$G$ and~$X$ and decrement~$k$ by one.
\end{rul}

\begin{lemma}\label{high-occurence-time-s}
  \autoref{high-occurence-fast-s} is correct. Given a solution~$X$ and the set~$M(v)$ constructed by \autoref{mark-neighbors-s} for each vertex~$v\in X$, we can exhaustively apply \autoref{high-occurence-fast-s} in $O(|X|n+m)$~time.
\end{lemma}
 
\begin{proof}
  For each vertex~$v$ in a solution~$X$, \autoref{mark-neighbors-s} adds at most~$2s$ vertices to~$M(v)$ for each found \FISG{}. If a vertex~$v\in X$ satisfies $|M(v)|>2sk$, then more than~$k$ \FISG{}s pairwisely intersect only in~$v$. According to \autoref{high-occurence}, we can delete~$v$ from~$G$ and decrement~$k$ by one. The running time can be shown analogously to \autoref{high-occurence-time}.
\end{proof}

\begin{korollar}\label{size-M-s}
  Let~$X$ be a solution for~$G$. For each~$v\in X$, let~$M(v)$ be the
  set constructed by \autoref{mark-neighbors-s}. After exhaustively
  applying \autoref{high-occurence-fast} to~$G$ and~$X$, the peripheral set
  $M:=\bigcup_{v\in X} M(v)$ contains at most $2sk|X|$ vertices.
\end{korollar}

\noindent Given a graph~$G$ and a solution~$X$, we can apply
\autoref{isolated} without any changes compared to the case~$s=2$. As
we have seen in \autoref{unmarked-bounds}, \autoref{sets-in-hm} and
\autoref{per} then bound the number of connected components
in~$G-X$. It is left to bound their sizes.  To this end, we only need
to slightly change \autoref{rul:goodconnected} and
\autoref{ruecken}. Recall that the connected components in~$G-X$ are
induced by sets in the collection~$\hv$. We start with a revision of
\autoref{rul:goodconnected}:
 
\begin{rul}\label{rul:goodconnected-s}
  Let~$X$ be a solution, let $\hp\in\hv$ and let $R(\hp)$ be a
  redundant subset of~$H$ as defined in \autoref{redundant-def}. If $|R(\hp)|> k+2s-1$, choose an arbitrary vertex
  from~$R(\hp)$ and remove it from~$G$.
\end{rul}

\noindent For the correctness proof of
\autoref{rul:goodconnected-s}, observe that \autoref{connected} and \autoref{in-hi} are still valid if we prove them under the following assumption instead of proving them under \autoref{working-asu}:

\begin{asu}\label{working-asu-s}
  Let~$X$ be a solution and let $R(\hp)$~be a redundant subset of ${\hp\in\hv}$. Assume that \autoref{rul:goodconnected-s} chooses to remove ${u\in R(\hp)}$ from~$G$.  Further, assume that there exists a solution~$S$ with $|S|\leq k$ for the graph~$G-\{u\}$.
\end{asu}

\noindent \autoref{working-asu-s} implies that $|R(\hp)\setminus(S\cup\{u\})|\geq 2s-1$; otherwise, \autoref{rul:goodconnected-s} could not have been applied. Because~$G[H]$ is a 2-plex, $G[R(\hp)\setminus(S\cup\{u\})]$~is connected. The graph $G[\hp\setminus(S\cup\{u\})]$ is connected for the same reason.  In the following, we write $G'$~for~$G-\{u\}$. 

\begin{lemma}\label{goodconnected-correct-s}
  \autoref{rul:goodconnected-s} is
  correct.
\end{lemma}

\begin{proof}
  We have to show that~$(G,k)$ is a yes-instance if and only
  if~$(G',k)$ is a yes-instance. If $(G,k)$ is a yes-instance, then there
  exists a solution~$S$ with~$|S|\leq k$ such that $G-S$~is an \pcg
  s. Clearly, then also $G'-S$~is an \pcg s and $(G',k)$~is a
  yes-instance.

  If $(G',k)$~is a yes-instance, then there exists a solution~$S$ with~$|S|\leq k$ such that $G'-S$~is an \pcg s, implying that \autoref{working-asu-s} is true. Assume that $G-S$~contains a \FISG{}. By \autoref{FISG-has-v}, there exists a \FISG{}~$F$ in~$G-S$ containing the vertex~$u$.  Because~$F$ is a \FISG{}, it contains a vertex~$v$ that is connected but nonadjacent to a set~$W$ of $s$ other vertices in~$F$.

  If~$u\notin\{v\}\cup W$, then \autoref{connected} shows that the
  vertices in~$\{v\}\cup W$ are connected to all vertices
  in~$\hp\setminus(S\cup\{u\})$. Thus, the vertices in~$\{u\}\cup W$
  would exist in~$G'-S$ and would be connected. That contradicts
  $G'-S$~being an \pcg s, because~$v$ is nonadjacent but connected to
  the~$s$ vertices in~$W$. Thus, we have $u\in\{v\}\cup W$.

  First, assume that $u=v$. That is, the vertex~$u\in R(\hp)$ is
  nonadjacent to~$W$. From \autoref{in-hi}, we can conclude
  that~$W\subseteq \hp\setminus S$. Because also~$u\in\hp\setminus S$,
  this contradicts the graph $G[\hp\setminus S]$ being an \pl s. Thus,
  we have~$u\in W$.

  Because~$u\in W$, we have that~$u$ is nonadjacent to~$v$. From \autoref{in-hi}, we conclude that $v\in\hp\setminus S$. By \autoref{redundant-def}, there exists an \md X $Z(H)$ with~$R(H)\subseteq Z(H)\subseteq H$ and~$v\in Z(H)$, because the vertex~$v\in H\setminus S$ is nonadjacent to the vertex~$u\in R(H)$. But then, because the vertex~$v\in Z(H)$ is nonadjacent to~$W$, the vertices in~$W$ must also be in~$H\setminus S$ by \autoref{in-hi}. Because also~$v$ is in $H\setminus S$, this again contradicts~$G[\hp\setminus S]$ being an \pl s. We conclude that $G-S$~must be a \pcg s.  Thus, $(G,k)$~is a yes-instance.
\end{proof}

\noindent We employ \autoref{construct-D} to construct redundant
sets. For a solution~$X$, the bound on the number
of vertices in a connected component in~$G-X$ then changes as follows:
 
\begin{lemma}\label{h1-vertices-s}
  Let the set~$M$ be peripheral with respect to a solution~$X$.  For
  a vertex set~${\hp\in\hv}$, let~$R(H)$ be the redundant subset constructed by
  \autoref{construct-D}. After exhaustively applying
  \autoref{rul:goodconnected-s} using~$R(H)$, the number of vertices
  in~$\hp\setminus M$ is $O(s|\hp\cap
  M|+s^2|\nx{\hp\setminus M}|+k)$.
\end{lemma}
 
\begin{proof}
  To prove the above lemma, we study the sets constructed in \autoref{construct-D}. By construction of~$R(\hp)$, we have~$R(\hp)=\hp\setminus\bar R(\hp)$. Observe that because~$\bar R(\hp)\subseteq \hp$, we also have~$\hp\setminus R(\hp)=\bar R(\hp)$. Because~$G[\hp]$ is an \pl s, there exist at most~$s-1$ vertices~$u\in\hp$ for every vertex~$w\in \hp\cap M$ such that~$u$ and~$w$ are nonadjacent. Thus, we have $|A(\hp)|\in O(s|\hp\cap M|)$.  Because~$M$ is peripheral, we can conclude from \autoref{per-s} that for each vertex~$w\in \nx{\hp\setminus M}$, there are at most ${2s-3}$~vertices~$u\in \hp$ such that~$u$ and~$w$ are nonadjacent. Thus, we have $|B(\hp)|\in O(s|\nx{\hp\setminus M}|)$.  Now, again because $G[\hp]$~is an \pl s, there exist at most $s-1$~vertices~$u\in\hp$ for every vertex~$w\in B(\hp)$ such that~$u$ and~$w$ are nonadjacent. Thus, we have $|C(\hp)|\in O(s|B(\hp)|)\subseteq O(s^2|\nx{\hp\setminus M}|)$.  This shows that the number of vertices in~$\hp\setminus (R(\hp)\cup M)$ is $O(s|\hp\cap M|+s^2|\nx{\hp\setminus M}|)$.  After applying \autoref{rul:goodconnected}, the number of vertices in~$R(H)$ is $O(k)$.
\end{proof}

\noindent Let~$M$ be peripheral with respect to a solution~$X$. We now revise \autoref{ruecken} to reduce the sizes of the connected components in~$G-X$ that are induced by sets~${H\in\hv}$ such that~$H\setminus M$ is nonadjacent to~$X$. Refer to \autoref{fig:h0} for an example.

\begin{rul}\label{ruecken-s}
  Let~$X$ be a solution and let~$H\in\hv$. Given a vertex set~$M$ such that~$H\setminus M$ is nonadjacent to~$X$, if $|\hp\setminus M|>|\hp\cap M|+2s-3$, then choose a vertex from $\hp\setminus M$ and remove it from~$G$.
\end{rul}

\begin{lemma}\label{ruecken-time-s}
  \autoref{ruecken-s} is correct. Given a vertex set~$M$, we can exhaustively apply \autoref{ruecken-s} in $O(n+m)$~time.
\end{lemma}

\begin{proof}
  Let~$S$ be a solution. First, observe that analogous to the proof of \autoref{ruecken-lemma}, we can show that if~$S$ does not destroy all edges between vertices in~$H$ and~$X$, then it must contain at least~$|\hp\setminus M|-(2s-3)$ vertices from~$\hp\setminus M$. If~$|\hp\setminus M|\geq|\hp\cap M|+2s-3$, then we can analogously to the proof of \autoref{exist-sol} find a solution~$S'$ with~$|S'|\leq|S|$ that destroys all edges between vertices in~$H$ and~$X$. From this, \autoref{ruecken-time-s} follows analogously to \autoref{ruecken-time}.
\end{proof}

\begin{korollar}\label{h0-vertices-s}
  Let~$X$ be a solution. Assume that there is a vertex set~$M$ and a set~${\hp\in\hv}$ such that~$H\setminus M$ is nonadjacent to~$X$.  After exhaustively applying \autoref{ruecken-s} given~$M$, the number of vertices in~$\hp\setminus M$ is $O(s+|\hp\cap M|)$.
\end{korollar}
 
\section{Kernel Size}
\label{result-s}
Given an \pvd s instance $(G,k)$, we now give a bound on the number of vertices in~$G$ after all data reduction rules have been applied. Given a solution~$X$, recall that for a connected component in~$G-X$ that is induced by a set~${H\in\hneg}$, the set~$H\setminus M$ is nonadjacent to~$X$ (cf.\,\autoref{fig:h0}); for a vertex set~${H\in\hpos}$, the set~$H\setminus M$ is adjacent to~$X$ (cf.\,\autoref{fig:h1}). We handle connected components induced by vertex sets in~$\hpos$ and~$\hneg$ separately.

\begin{lemma}\label{nx-sum-s}
Let the set~$M$ be peripheral with respect to a solution~$X$.
For the sets
  \begin{align*}
    X_1 &:=\{v\in X\mid\text{there is exactly one set~$\hp\in\hv$ such that~$\hp\setminus M$ is adjacent to~$v$}\}\text{,}\\
    X_2 &:= X\setminus X_1\text{ and}\\
    \mathcal{\tilde H} &:= \{\hp\in\hpos\mid \hp\setminus M \text{ is adjacent to only vertices in }X_1\}\text{,}
  \end{align*}
  the following relations hold:
  \begin{align*}
    \sum_{\hp\in\mathcal{\tilde H}}|\nx{\hp\setminus M}|&=|X_1| &&\text{and}& |\mathcal{\tilde H}|&\leq|X_1| &&\text{and}& |\hpos\setminus\mathcal{\tilde H}|&\leq s|X_2|.
  \end{align*}
\end{lemma}
 
\begin{proof}
  This follows analogously to the proof of \autoref{nx-sum} with \autoref{per-s} for peripheral sets.
\end{proof}

\begin{satz}
  \pvd s has a problem kernel with~$O(k^2s^3)$ vertices. It can be found in $O(ksn^2)$~time.
\end{satz}

\begin{proof}
  Given an \pvd s instance $(G,k)$, we first find a constant-factor approximate solution~$X$ for~$G$ using \autoref{find-X}. If~$(G,k)$ is a yes-instance, we have~$|X|\in O(sk)$ according to \autoref{X-boundary-s}. In this case, we find~$|X|$ in $O(ks(n+m))$ time according to \autoref{find-X-time-s}, because if we find more than~$k$ \FISG{}s using \autoref{find-X}, then we can stop and output that~$(G,k)$ is a no-instance. After constructing the constant-factor approximate solution~$X$, we construct a set~$M$ that is peripheral with respect to~$X$. According to \autoref{mark-neighbors-time-s}, this can be done in~$O(|X|n^2)$ time using \autoref{mark-neighbors-s}. According to \autoref{size-M-s}, we can use \autoref{high-occurence-fast-s} to reduce the size of~$M$ to at most $2sk|X|$ vertices. This can be done in~$O(|X|n+m)$ time according to \autoref{high-occurence-time-s}. We then apply \autoref{isolated} in~$O(n+m)$~time as shown in \autoref{isolated-time}, followed by \autoref{rul:goodconnected-s}. Analogously to the proof of \autoref{goodconnected-time}, we can show that this works in~$O(n^2)$~time.  Finally, we apply \autoref{ruecken-s}, which runs in~$O(n+m)$ time; this follows analogously to the proof of \autoref{ruecken-time}.

  We now count the vertices that remain in~$G$. The graph~$G$ contains vertices from~$X$, vertices from~$M$, and vertices from the connected components in~$G-X$ that are not in~$M$. As shown above, we have~$|X|\in O(sk)$ and~$|M|\in O(s^2k^2)$. It is left to count the vertices in $\bigcup_{\hp\in\hv} (\hp\setminus M)$. Let~$\mathcal{\tilde H}, X_1,$ and~$X_2$ be as defined in \autoref{nx-sum-s}. We can conclude from \autoref{h1-vertices-s} and \autoref{h0-vertices-s} that the size of $\bigcup_{\hp\in\hv} (\hp\setminus M)$ is
  \begin{align*}
    &O\big(\smashoperator{\sum_{\hp\in\hpos}}(s |\hp\cap
    M|+s^2|\nx{\hp\setminus
      M}|+k)+\smashoperator{\sum_{\hp\in\hneg}}(|\hp\cap M|+s)\big).\\
    \intertext{Because the sets in~$\hv$ are pairwise disjoint, we have that $\sum_{H\in\hv}s|H\cap M|\leq s|M|$. Thus, the above term is}
    & O\big(\smashoperator{\sum_{\hp\in\hpos}}(s^2|\nx{\hp\setminus M}|+k)+s|M|+s|\hneg|\big).\\
    \intertext{By \autoref{sets-in-hm}, we have that~$|\hneg|\leq
      |M|$. Thus, this is}
  & O\big(\smashoperator{\sum_{\hp\in\hpos}}(s^2|\nx{\hp\setminus M}|+k)+s|M|\big).\\
    \intertext{For each set~$H\in\hpos\setminus\mathcal{\tilde H}$,
      the set~$H\setminus M$ must be adjacent to a vertex
      from~$X_2$. This is by definition of~$\mathcal{\tilde H}$ in
      \autoref{nx-sum-s} and by definition of $\hpos$. From \autoref{per-s}, we can conclude
      that~$|H\setminus M|\in O(s)$, implying that only the sets in~$\mathcal{\tilde H}$ may contain~$\Theta(s^2|\nx{\hp\setminus M}|+k)$ vertices that are not in~$M$. Thus, we have}
  \big|\smashoperator{\bigcup_{H\in\hv}}(H\setminus M)\big|&\in O\big(\smashoperator{\sum_{\hp\in\mathcal{\tilde H}}}\big(s^2|\nx{\hp\setminus
  M}|+k\big)+s|\hpos\setminus\mathcal{\tilde H}|+s|M|\big).
\end{align*}
By \autoref{nx-sum-s}, this is $O(|X_1|s^2+k|X_1|+|X_2|s^2+s|M|)$. Using $|X_1|+|X_2|=|X|$ and adding the vertices in~$M$ and~$X$, this is $O((s^2+k)|X|+s|M|)$. Thus, the total number of vertices in~$G$ is~$O(k^2s^3)$.
\end{proof}

\chapter{Conclusion and Outlook}
We have shown an $O(k^2s^3)$-vertex problem kernel for \pvd s. This result is comparable with the~$O(ks^2)$-vertex problem kernel for \name{$s$-Plex Editing} shown by Guo et al.\cite{DBLP:conf/aaim/GuoKNU09}: in an~$n$-vertex graph, one vertex deletion can lead to~$n-1$ edge deletions. Under the assumption that input graphs for clustering problems are typically dense, this suggests that typical parameter values for \name{$s$-Plex Editing} are at least quadratic in parameter values for \pvd s; the parameter is the number of allowed graph modifications. Seen from this angle, our result seems consistent with the result that \name{$s$-Plex Editing} has a problem kernel with~$O(ks^2)$~vertices.

It is open whether \pvd s has an $O(ks^c)$-vertex problem kernel for some constant~$c$. It is also open to improve the $s^3$-factor in the number of vertices in our problem kernel. This factor results from the size of the constant-factor approximate solution shown in \autoref{X-boundary-s}, from the size of the peripheral set shown in \autoref{size-M-s}, and from the way we construct redundant sets in \autoref{h1-vertices-s}. The most promising approach to improve on the~$s^3$-factor seems to be the construction of larger redundant sets so that more vertices can be removed by \autoref{rul:goodconnected-s}.

In \autoref{introsec}, we discussed how to solve \pvd s using \hs d for a natural number~$d\in O(s+\sqrt s)$. For \hs d, problem kernels containing~$O(k^{d-1})$ or~$O(k^d)$ elements are known\cite{DBLP:conf/wads/Abu-Khzam07,DBLP:conf/stacs/Kratsch09,Flu06}. This bound is exponential in~$d$ and $d$ is in turn bounded by a function linear in~$s$. This yields an upper bound on the number of elements in a \hs d kernel that is exponential in~$s$. From this angle, it is remarkable that problem kernels for \pvd s as well as for \name{$s$-Plex Editing} exist whose number of vertices is bounded by a polynomial in~$s$ as~well~as~in~$k$.

It might be hard to find a search tree for \pvd s that is smaller than the search tree for an equivalent \hs d instance. However, a \hs d instance obtained from our \pvd s problem kernel contains~$O(k^{2d})$ sets. This bound is exponential in~$d$. Thus, constructing a \hs d instance from an \pvd s instance might be practically infeasible. It is open to find faster algorithms for \pvd s that do not rely on \hs d.

The most promising approach to faster algorithms for \pvd s seems to be iterative compression\cite{GMN2009} introduced by Reed et al.\cite{RSV04}.  Hüffner et al.\cite{ HKMN09TOCS} have successfully applied it to \cvd. Using iterative compression, we can solve \pvd s by solving multiple instances of the following problem:

\decprob{\textsc{Disjoint} \pvd s}{A graph~$G=(V,E)$, a non-negative number~$k$, and a solution~$S\subseteq V$ with~$|S|\leq k+1$ such that~$G-S$ is an \pcg s.}{Is there an alternative solution~$S'$ with~$S\cap S'=\emptyset$ and~$|S'|\leq k$ such that~$G-S'$ is an \pcg s?}

\noindent Fellows et al.\cite{DBLP:conf/mfcs/FellowsGMN09} have shown that while the analog problem for \cvd{} is in~$\DP$, \textsc{Disjoint} \pvd s is $\NP$-hard.  After some initial observations on the case~$s=2$, we guess that \textsc{Disjoint} \pvd 2 can be solved using a size-$O(2^k)$ search tree. In combination with our kernelization algorithm, we could then solve \pvd 2 in~${O(3^kk^c+kn^2)}$ time for some constant~$c$.

\bibliographystyle{plain}
\bibliography{splex-notizen}{}

\end{document}

